 \nonstopmode

\newcommand{\ITRS}[2]{#1} \newcommand{\Ignore}[2]{#2}  

\newif\ifpaper \paperfalse

\ITRS{
\documentclass[submission,copyright,creativecommons]{eptcs}
\providecommand{\event}{ITRS '16} 

\usepackage{etex}
\usepackage{relsize}

\usepackage{eptcsconversion}

\newif\ifmycolour \mycolourfalse
}{	\documentclass{Paper}
}

\def\sk{\,}

\usepackage{Xmacros}
\def \TurnLmu{\Turn}

\ITRS{
 \let\oldproof\proof
 \let\oldendproof\endproof

 \proofcorrection-12\point

 \renewenvironment{proof}{\collectbody\Comment}{}
}{\newenvironment{Proof}{\begin{proof}}{\end{proof}}}

 \def \ssubscr{\hbox{\relsize{-2}{\sc s}}}
 \def \esubscr{\hbox{\relsize{-2}{\sc e}}}

 \def \seqS {\mathbin{\leq_{\ssubscr}}}

 \def \reduces {\mathrel{\rightarrow}}
 \def \betamusubscript {_{`b`m}}
 \def \redbmu {\mathbin{\semcolour \arr \betamusubscript}}
 \renewcommand{\rtcredbmu}[1][\ast]{\mathrel{\comparr{#1}{\betamusubscript}}}

 \def \eqbmu {\mathbin{\semcolour =\betamusubscript}}

 \def \dcol{\mathrel{::}}
 
 \def \App {\textit{App}}
 \def \Abs {\textit{Abs}}
 
 \def \TurnLmuInt {\mathbin{{\vdash}\kern-2\point_{\wedge}}}
 
 \def \TurnLmu {\mathbin{{\vdash}}}
 \def \TurnLmuBot {\mathbin{{\vdash}\kern-2\point_{\perp}}}
 \def \derlmuBot{ \def \TurnLmu {\TurnLmuBot} \derLmu }

 \def \TurnlmuS {\mathbin{{\vdash}\kern-2\point_{\ssubscr}}}
 \def \notTurnlmuS {\mathbin{{\not\vdash}\kern-2\point_{\ssubscr}}}
 \def \derlmuS{ \def \TurnLmu {\TurnlmuS} \derLmu }

 \def \TurnlmuSN {\mathbin{{\vdash}\kern-2\point_{\hbox{\scriptsize \sc sn}}}}
 \def \derlmuSN{ \def \TurnLmu {\TurnlmuSN} \derLmu }

 \def \TurnlmuE {\mathbin{{\vdash}\kern-2\point_{\esubscr}}}
 \def \notTurnlmuE {\mathbin{{\not\vdash}\kern-2\point_{\esubscr}}}
 \def \derlmuE{ \def \TurnLmu {\TurnlmuE} \derLmu }

 \def \TurnlmuIU {\mathbin{{\vdash}\kern-2\point_{\cap\cup}}}
 
 \def \derlmu{\derLmu}
 \def \TurnLmu {\Turn}
 \def \Rename{\textsc{Ren}}
 \def \Top{`w}

 \def\interBot{\,{\inter}_{\perp}}
 
 \def \BuildInt(#1){{\cap}_{#1}\qskp}

 \def \BuildInt({\ifnextchar{I}{\BuildIntI(}{\ifnextchar{J}{\BuildIntJ(}{\ifnextchar{n}{\BuildIntn(}{\ifnextchar{m}{\BuildIntm(}{\BuildIntAux(}}}}} 
 \def \BuildIntAux(#1)#2{\int{#1}#2_i}
 \def \BuildIntI(#1)#2{\int{I}#2_i}
 \def \BuildIntJ(#1)#2{\int{J}#2_j}
 \def \BuildIntn(#1)#2{\int{\n}#2_i}
 \def \BuildIntm(#1)#2{\int{\m}#2_j}

 \def \CoI#1{\BuildInt(#1){`C}}
 \def \DoI#1{\BuildInt(#1){\cont{D}}}

 \def \derred {\mathrel{\arr _{\mbox{\scriptsize \sc Der}}}}
 \def \rtcderred {\mathrel{\arr ^{*}_{\mbox{\scriptsize \sc Der}}}}

 \def \SN(#1){\textsl{SN}\psk(#1)}

 \def \CompBr(#1){\textsl{Comp}\,(#1)}

 \def \MojM#1{\join_{\ul{#1}}M_i}
 \def\Appr #1 |- #2 : #3 | #4 {\Exists A \ele \SetAppr(#2) \Pred[\derlmuS #1 |- A : #3 | #4 ]}

 \def \Redex{\mbox{\textsf{R}}}

 \def \LorBr(#1){\textit{lor}\psk(#1)}
 \def \LorNoBr#1{\textit{lor}\psk(#1)}
 \def \Lor{\ifnextchar({\LorBr}{\LorNoBr}}

 \def \lored{\mathrel{\rightarrow_{\textit{\scriptsize lor}}}}
 \def \redsize(#1){\#\,{#1}}

\title{Characterisation of Approximation and \\ (Head) Normalisation for $`l`m$ ~\\ using Strict Intersection Types}
\author{Steffen van Bakel
\institute{Department of Computing, Imperial College London, 
180 Queen's Gate, London SW7 2BZ, UK}
}
\date{Extended Abstract}

 \def\titlerunning{Approximation and (Head) Normalisation for $`l`m$ using Strict Intersection Types}
 \def\authorrunning{Steffen van Bakel}

\begin{document}

\ITRS{\maketitle}{} 

 \begin{abstract}
We study the strict type assignment for $\lmu$ that is presented in \cite{Bakel-FSCD'16}.
We define a notion of approximants of $\lmu$-terms, show that it generates a semantics, and that for each typeable term there is an approximant that has the same type.
We show that this leads to a characterisation via assignable types for all terms that have a head normal form, and to one for all terms that have a normal form, as well as to one for all terms that are strongly normalisable.
 \end{abstract}

\ITRS{}{\maketitle} 
 
 \section*{Introduction}

The Intersection Type Discipline \cite{BCD'83} is an extension of the standard, implicative type assignment known as Curry's system \cite{Curry-Feys'58} for the $`l$-calculus \cite{Church'36,Barendregt'84}; the extension made consists of relaxing the requirement that a parameter for a function should have a single type, adding the type constructor $\inter$ next to $\arrow$.
This simple extension allows for a great leap in complexity: not only can a (filter) model be built for the {\LC} using intersection types, also strong normalisation (termination) can be characterised via assignable types; however, type assignment becomes undecidable.

A natural question is whether or intersection type assignment yields a semantics also for other calculi, like $\lmu$ \cite{Parigot'92}.
To answer that, in \cite{Bakel-Barbanera-Liguoro-TLCA'11,BakBdL-ITRS12,Bakel-Barbanera-Liguoro-LMCS'15} a notion of intersection type assignment was defined for $\lmu$ that is a variant of the union-intersection system defined in \cite{Bakel-ITRS'10}.
Inspired by Streicher and Reus's domain \cite{Streicher-Reus'98}, $\lmu$-terms are separated into terms and \emph{streams}; then $\lmu$'s names act as the destination of streams, the same way variables are the destination of terms.
A type theory is defined following the domain construction; 
the main results for that system are the definition of a filter model, closure under conversion, and that the system is an extension of Parigot's \cite{Bakel-Barbanera-Liguoro-TLCA'11}; and that, in a restricted system, the terms that are typeable are exactly the strongly normalising ones \cite{BakBdL-ITRS12}.

One of the main disadvantages of taking the domain-directed approach to type assignment is that, naturally, intersection becomes a `top level' type constructor, that lives at the same level as arrow, for example, which induces a contra-variant type inclusion relation `$\seq$' and type assignment rule $(\seq)$ that greatly hinder proofs and gives an intricate generation lemma.
This problem is addressed in \cite{Bakel-FSCD'16} where a \emph{strict} version of the system of \cite{Bakel-Barbanera-Liguoro-LMCS'15} is defined, in the spirit of that of \cite{Bakel-TCS'92,Bakel-ACM'11} that allows for more easily constructed proofs.
The main restriction with respect to the system of \cite{Bakel-Barbanera-Liguoro-LMCS'15} is limiting the type  inclusion relation to a relation that is no longer contra-variant, and allows only for the selection of a component of an intersection type; this is accompanied by a restriction of the type language, essentially no longer allowing intersection on the right of an arrow.
The main results shown in \cite{Bakel-FSCD'16} are that the system is closed under conversion (i.e.~under reduction and expansion), and that all terms typeable in a system that excludes the type constant $`w$ are strongly normalisable.
To that aim it shows that, in this system, cut-elimination is strongly normalisable, using the technique of derivation reduction \cite{Bakel-NDJFL'04} (see also \cite{Bakel-TCS'08,Bakel-ACM'11}).

In this paper, we will elaborate further on the strict system. 
As in \cite{Bakel-TCS'08,Bakel-ACM'11}, in this paper we will show that the fact that derivation reduction is strongly normalisable also here leads to an approximation result. 
For that, we define a notion of approximation for $\lmu$, and show that this yields a semantics (Thm.\,\ref{approx seman lmux}).
We then show that for every typeable term there exists an approximant of that term that can be assigned exactly the same types (Thm.\,\ref{approximation result for TurnlmuS}).
We then show that this approximation result naturally gives a characterisation of head normalisation (Thm\,\ref{characterisation of head normal form}), as well as a characterisation of normalisation (Thm\,\ref{characterisation of normalisation}).
We also revisit the proof of characterisation of strong normalisation of terms through the assignable types (Thm\,\ref{SN theorem}), which thanks to the approximation result has a more elegant proof. 

Because of the restricted available space, most of the (full) proofs are not presented here.
A version of this paper with the proofs added in an appendix can be found at \url{www.doc.ic.ac.uk/~svb/Research/Papers/ITRS16wapp.pdf}. 

\emptyline
\mysfbf{Note:}
We will write $\n$ for the set $\{1,\ldots,n\}$ and use a vector notation for the abbreviation of sequences, so 
write $\Vect[n]{X}$ for $X_1, \ldots, X_n$, and $\Vect{X}$ if the number of elements in the sequence is not important. 

 \section{The $`l`m$-calculus} \label{lambda mu section}

In this section we present Parigot's pure $`l`m$-calculus as introduced in \cite{Parigot'92}.
It is an extension of the untyped $`l$-calculus obtained by adding \emph{names} and a name-abstraction operator $`m$ and was intended as a proof calculus for a fragment of classical logic. 
Derivable statements have the shape $\derLmu `G |- M : A | `D $, where $A$ is the main (\emph{active}) conclusion of the statement, and $`D $ contains the alternative conclusions, consisting of pairs of names and types; the left-hand context $`G$, as usual, is a mapping from term variables to types, and represents the assumptions about free variables of $M$.

 \begin{definition} [Term Syntax \cite{Parigot'92}] \label{def:terms}
Let $x,y,z,\ldots $ range over \emph{term variables}, and $`a,`b,`g,`d,\ldots$ range over  \emph{names}.
The \emph{terms}, ranged over by $M,N,P,Q,\ldots$ are defined by the grammar:
 \[ \begin{array}{rcl}
M,N & ::= &x \mid `l y.M \mid MN \mid `m`a.[`b]M
 \end{array} \]

 \end{definition}

As usual, we consider $`l$ and $`m$ to be binders; the sets $\fv(M)$ and $\fn(M)$ of, respectively, \emph{free variables} and \emph{free names} in a term $M$ are defined in the usual way.
We adopt Barendregt's convention on terms, and will assume that free and bound variables and names are different.

 \begin{definition} [Substitution \cite{Parigot'92}] \label{def:substitution}
Substitution takes two forms:%
 \[ \begin{array}{l@{\hspace{2mm}}cll}
 \textit{term substitution:} & M[N/x] & \textrm{($N$ is substituted for $x$ in $M$)}
 \\
 \textit{structural substitution:} & M [L{`.}`g / `a] &
			\textrm{(every `subterm' $[`a]N$ of $M$ is replaced by $[`g]NL$)}
 \end{array} \]
\Comment{
More precisely, $M [L{`.}`g / `a]$ is defined by:
 \[ \begin{array}{rcll}
(`m`b.[`a]P) [L{`.}`g / `a] &\ByDef& `m`b.[`g]`( P [L{`.}`g / `a] `) \, L
	\\ [1mm]
(`m`b.[`d]P) [L{`.}`g / `a] & \ByDef & `m`b.[`d ] P [L{`.}`g / `a] \dquad \textrm{if $`a\neq`d $}
	\\ [1mm]
x [L{`.}`g / `a] & \ByDef & x
	\\ [1mm]
(`l x.P) [L{`.}`g / `a] & \ByDef & `l x.P [L{`.}`g / `a]
	\\ [1mm]
(PQ) [L{`.}`g / `a] & \ByDef & ({P [L{`.}`g / `a] }) \, ({ Q [L{`.}`g / `a] })
 \end{array} \]
}As usual, both substitutions are capture avoiding, using $`a$-conversion when necessary.

\Comment{
Notice that, in the first two alternatives, by Barendregt's convention we can assume that $`b \not= `a$.
}

 \end{definition}

 \begin{definition}[Reduction \cite{Parigot'92}]\label{def:reduction}
Reduction in $\lmu$ is based on the following rules:
 \[ \begin{array}{r@{\dquad\dquad}rll@{\dquad}l}
(`b ): & (`l x.M)N & \reduces & M[N/x] & (\textit{logical reduction}) \\
(`m)\footnotemark
: & \multicolumn{3}{l}{\kern-8.25mm
\begin{cases}{@{}rcll}
(`m`b .[`b]P)Q &\reduces& `m`g . [`g]`(P[Q{`.}`g / `b]`)Q \\
(`m`b .[`d]P)Q &\reduces& `m`g . [`d]P[Q{`.}`g / `b], & \textrm{if } `d \not= `b \\
 \end{cases}
} & (\textit{structural reduction}) \\
(\Rename): &
\multicolumn{3}{l}{\kern-13.75mm
\begin{cases}{@{}rcl@{\dquad}l}
`m`a.[`b]`m`g .[`g]M & \reduces & `m`a.[`b]M[`b/`g] \\
`m`a.[`b]`m`g .[`d]M & \reduces & `m`a.[`d]M[`b/`g], & \hspace*{6.25mm} \textrm{if } `d \not= `g
 \end{cases}
} & (\textit{renaming}) 
 \end{array} \]
\footnotetext
{
A more common notation for the second rule, for example, would be $(`m`b.[`d]M)N \reduces `m`b.[`d]M [N / `b]$.
This implicitly uses the fact that $`b$ disappears during reduction, and through $`a$-conversion can be picked as name for the newly created applications instead of $`g$.
But, in fact, this is not the same $`b$ (and the named term has changed), as reflected in the fact that its type changes during reduction. 
Moreover, when making the substitution \emph{explicit} as in \cite{Bakel-Vigliotti-CLaC'14}, it becomes clear that this other approach in fact is a short-cut, which our definition does without. }
We write $M \redbmu N$ for the reduction relation that is the compatible closure of these rules, and $\eqbmu$ for the equivalence relation generated by it.
 \end{definition}
Confluence for this notion of reduction has been shown in \cite{Py-PhD'98}.

We will need the concept of head-normal form for $\lmu$, which is defined as follows:

 \begin{definition}[Head-normal forms]
The $\lmu$ \emph{head-normal forms} 
(with respect to our notion of reduction $\redbmu$) are defined through the grammar:%
 \[ \begin {array}{rrl@{\quad}l}
\lmuHNF & ::= &
xM_1\dots M_n & (n \geq 0)
	\\ &\mid &
`lx. \lmuHNF 
	\\ &\mid &
\muterm`a.[`b] \lmuHNF
	& (\lmuHNF \not= \muterm`g.[`d] \lmuHNF' )
 \end {array} \]
 \end{definition}

 \section{Strict type assignment} \label{Strict type assignment}

Intersection (and union) type assignment for $\lmu$ was first defined in \cite{Bakel-ITRS'10}; this was followed by \cite{Bakel-Barbanera-Liguoro-TLCA'11}, in which an intersection type theory is developed departing from Streicher and Reus's domain construction \cite{Streicher-Reus'98}.
Terms can be typed with functional types $`d$ and streams by continuation types $`k$ that are of the shape $`d_1\prod \dots \prod `d_n \prod `w$, so essentially is a sequence of $`d$s. 
This later \cite{BakBdL-ITRS12} was followed by the proof that, as for the $`l$-calculus, the underlying intersection type system for $\lmu$ allows for the full characterisation of strongly normalisable terms; in that paper, renaming is not considered.
These papers were later combined (and revised) into \cite{Bakel-Barbanera-Liguoro-LMCS'15}.
One of the main disadvantages of taking the domain-directed approach to type assignment is that, naturally, intersection becomes a `top level' type constructor, that lives at the same level as arrow, for example.
This in itself is not negative, since it gives readable types and easy-to-understand type assignment rules, but it also induces a contra-variant type inclusion relation `$\seq$' and type assignment rule $(\seq)$ that 
hinder proofs and give an intricate generation lemma (see \cite{Bakel-Barbanera-Liguoro-LMCS'15} for details).

Therefore, in \cite{Bakel-FSCD'16}, a strict restriction of the system of \cite{Bakel-Barbanera-Liguoro-LMCS'15} was presented, where the occurrence of intersections is limited to only appear as components of continuation types (so no intersections of continuation types), and type inclusion is no longer contra-variant and only allows for the selection of a component in an intersection type.
It also uses $`W$ rather than $`w$ to mark the end of a continuation type. 
But, more importantly, it removed the inference rule $(\seq)$, and changed the type assignment rules to explicitly state when a $\seq$-step is allowed, as in rule $(\Ax)$.

This system is defined as follows:

 \begin{definition}[Strict Types \cite{Bakel-FSCD'16}] \label{strict types}
 \begin{enumerate}
 \firstitem
Let $`y$ range over a countable, infinite set of type constants.
We define our strict types by the grammar:
 \[ \begin{array}{rcl@{\quad}l@{\quad}l}
`A,`B &::=& `C\arr `y && \textit{basic types} \\
`R,`S,`T &::=& `w \mid `A_1 \inter \dots \inter `A_n & (n \geq 1) & \textit{intersection types} \\
`C,\cont{D} &::=& `W \mid `S\prod `C 
&& \textit{continuation types}
 \end{array} \]

 \item
On strict types, the type inclusion relation $\seqS$ is the smallest partial order satisfying the rules:%
\def\colw{\hspace{6mm}}
 \[ \begin{array}{c@{\colw}c@{\colw}c@{\colw}c@{\colw}c}
\Inf	[j \ele \n, ~ n \geq 1]
	{ `A_1 \inter \dots \inter `A_n \seq `A_j }
	&
\Inf	[n \geq 1]
	{ `S \seq `A_i \quad (\forall i \ele \n)
	}{ `S \seq `A_1 \inter \dots \inter `A_n }
	&
\Inf	{ `S \seq `w }
	&
\Inf	{ `C \seq `W}
	&
\Inf	{`S \seq`T \quad `C \seq \cont{D}
	}{ `S\prod `C \seq `T\prod \cont{D} }
 \end{array} \]

 \end{enumerate}
 \end{definition}

For convenience, we will write $\AoI{I}$ for $`A_{i_1} \inter \dots \inter `A_{i_n}$ where $I = \Set{i_1, \ldots, i_n}$, $\AoI{\emptyset}$ for $`w$, so the second and third rule combine to 
\[
\Inf	[n \geq 0]
	{ `S \seq `A_i \quad (\forall i \ele \n)
	}{ `S \seq `A_1 \inter \dots \inter `A_n }
\]
and $\AoI{\n}$ for $`A_1 \inter \ldots \inter`A_n$.
Notice that for any continuation type $`C$ there are $n \geq 0$ and $`S_i ~ (i \ele \n)$ such that $`C = `S_1\prod \dots `S_n \prod `W$.

 \begin{definition}[Strict Type Assignment \cite{Bakel-FSCD'16}] \label{strict type assignment for lmu}
 \begin{enumerate}

 \firstitem
A {\em variable context} $`G$ is a mapping from term variables to intersection types, denoted as a finite set of {\em statements} $x{:}`S$, such that the {\em subject} of the statements ($x$) are distinct.

 \item
We write $`G,x{:}`S$ for the context defined by:
\[
\begin {array}{rcll}
`G, x{:}`S & \ByDef & `G \Union \{x{:}`S\}, & \textrm{if $`G$ is not defined on $x$}
\\
	& \ByDef & `G, & \textrm{if }x{:}`S \ele `G 
\end {array} %
\]
We write $x \notele `G$ if there exists no $`S$ such that $x{:}`S \ele `G$.

 \item
\emph{Name contexts} $`D$ and the notions $`a{:}`C,`D$ and $`a \notele `D$ are defined in a similar way.

 \item
We define \emph{strict type assignment} for $`l`m$-terms through the following natural deduction system:
{\def \TurnlmuS {\Turn}%
 \[ \begin{array}[t]{@{}rl@{~\quad}rl}
(\Ax) : &
\Inf	[`S \seqS `A ]
	{\derlmuS `G, x{:}`S |- x : `A | `D }
&
(\inter) : &
\Inf	[I = \emptyset \Or |I| \geq 2]
	{\derlmuS `G |- M : `A_i | `D
	\quad
	(\Forall i \ele I)
	}{ \derlmuS `G |- M : \AoI{I} | `D }
 \\ [5mm]
(\Abs) : &
\Inf	[x \notele `G]
	{\derlmuS `G,x{:}`S |- M : `C\arr `y | `D
	}{\derlmuS `G |- `lx.M : `S\prod `C\arr `y | `D }
&
(`m) : &
\Inf	[`a \notele `D, `C \seqS \cont{D}]
	{\derlmuS `G |- M : \cont{D}\arr `y | `a{:}`C,`D
	}{\derlmuS `G |- `m`a.[`a]M : `C\arr `y | `D }
 \\ [5mm]
(\App) : &
\Inf	{\derlmuS `G |- M : `S\prod `C\arr `y | `D
	\quad
	\derlmuS `G |- N : `S | `D
	}{ \derlmuS `G |- MN : `C\arr `y | `D }
&
(`m') : &
\Inf	[<16mm>{`b \not= `a \And `a \notele `D, \\ `C' \seqS \cont{D}}]
	{\derlmuS `G |- M : \cont{D}\arr `y | `a{:}`C,`b{:}`C',`D
	}{\derlmuS `G |- `m`a.[`b]M : `C\arr `y | `b{:}`C',`D }
 \end{array} \]}
\noindent
We write $ \derlmuS `G |- M : `S | `D $ for judgements derivable using these rules, and prefix this with `$\D \dcol {}$' if we want to name the derivation.

 \item
The relation $\seqS$ is naturally extended to variable contexts as follows:
 \[ \begin{array}{rcl}
`G \seqS `G' & \ByDef & \Forall x{:}`S\ele `G' ~ \Exists x{:}`T\ele `G \Pred [ `T \seqS `S ];
 \end{array} \]
$ `D \seqS `D'$ is defined similarly. 

 \end{enumerate}
 \end{definition}
 
 \begin{definition}
By abuse of notation, we allow the notation $`S \inter `T$, where $`S = \AoI{n}$ and $`T = \BoI{m}$, which stands for $`A_1 \inter \dots \inter `A_n \inter `B_1 \inter \dots \inter `B_m$.
Given two contexts $`G_1$ and $`G_2$, we 
define the context $`G_1\inter `G_2$ as follows:
 \[ \begin{array}{rcr@{~}l}
`G_1\inter `G_2 &\ByDef& \Set{x{:}`S_1\inter `S_2 \mid x{:}`S_1 \ele `G_1 \And x{:}`S_2 \ele `G_2 } & \Union \\
&& \Set {x{:}`S \mid x{:}`S \ele `G_1 \And x \notele `G_2 } & \Union
\Set {x{:}`S \mid x{:}`S \ele `G_2 \And x \notele `G_1 }
 \end{array} \]
and write $\int_{\n}{`G_i}$ for $ `G_1\inter \dots \inter `G_n $.
We will also allow intersection of continuation types as short-hand notation: let $\cont{D} = `S_1 \prod \dots \prod `S_n$ $\prod `W $, and $`C = `T_1 \prod \dots \prod `T_m \prod `W $ and assume, 
that $n < m$; we define %
 \[ \begin{array}{rcl}
\cont{D}\inter `C &\ByDef& `S_1\inter`T_1 \prod \dots \prod `S_n \inter `T_n \prod `T_{n+1} \prod \dots \prod `T_m \prod `W . 
 \end{array} \] 
(we need this notion in the proof of Thm.~\ref{characterisation of head normal form}). 
Then $`D_1 \inter `D_2$ is defined the same way as $`G_1\inter `G_2$.

 \end{definition}

In \cite{Bakel-FSCD'16} it is then shown that this notion of type assignment is closed under conversion, so can be used to define a (filter) semantics.
That paper also defines a notion of cut-elimination, by defining derivation reduction $\derred$, where only those redexes in terms are contracted that are typed with a type different from $`w$; it shows that this notion is strongly normalisable, which then leads to the proof that all terms typeable in a restriction of $\TurnlmuS$ that eliminates the type constant $`w$, are strongly normalisable.

The main results shown in \cite{Bakel-FSCD'16} that are relevant to this paper are:

 \begin {theorem} [\cite{Bakel-FSCD'16}] \label{FSCD lemma}
 \begin{enumerate} 
 
 \firstitem 
\label{seqS admissible} \label{closed for seqS}
If $ \derlmuS `G |- M : `S | `D $, $`G' \seqS `G$, $`D' \seqS `D$,\footnote{
The condition $`D' \seqS `D$ 
might seem counterintuitive, since one might expect the inclusion relation to be reversed. 
To support intuition, we can see types in name contexts as negations, and $`a{:}`A\prod `W $ as $`a{:}\neg`A$.
Notice that $ `A\inter`B\prod `W \seqS `A\prod `W $
; obviously we have $`a{:}`A\inter`B\prod `W \seqS `a{:}`A\prod `W $ and $\neg `A \seq \neg`A \union \neg`B $.
}
and $`S \seqS `T$, then $ \derlmuS `G' |- M : `T | `D' $.
 
 \item \label {closed for eq}
If $ \derlmuS `G |- M : `A | `D $ and $M \eqbmu N$, then $ \derlmuS `G |- N : `A | `D $.

 \item \label {relation bred derred}
Let $\D \dcol \derlmuS `G |- M : `S | `D $, and $\D \rtcderred \D'\dcol \derlmuS `G |- N : `S | `D $, then $M\rtcredbmu N$.

 \item \label {typeable implies SN}
If $\D \dcol \derlmuS `G |- M : `S | `D $, then $\SN (\D)$ ($\D$ is strongly normalisable).
 \end{enumerate}

 \end{theorem}

 \section{Approximation semantics for \texorpdfstring{$\lmu$}{}}
\label{approximation section}

Following the approach of \cite{Wadsworth'76}, we now define an \emph{approximation semantics} for $\lmu$ with respect to $\redbmu$.

Essentially, approximants are partially evaluated expressions in which the locations of incomplete evaluation (\emph{i.e.}~where reduction \emph{may} still take place) are explicitly marked by the element $\bottom$; thus, they \emph{approximate} the result of computations.

Approximation for $\Lmu$ (a variant of $\lmu$ where naming and $`m$-binding are separated \cite{deGroote'94}) has been studied by others as well \cite{Saurin'10,deLiguoro'16}; \emph{weak} approximants for $\lmu$ are studied in \cite{Bakel-Vigliotti-CLaC'14}.

 \begin{definition}[Approximation for $\lmu$] \label{approximation lmu} \label{approximant definition} \label{definition approximate normal forms}
 \begin{enumerate}

 \firstitem
We define $\lmu\bottom$ as an extension of $\lmu$ by adding the term constant $\bottom$.

 \item
The set of $\lmu$'s \emph{approximants} $\SetAppr$ with respect to $\redbmu$ is defined through the grammar:%
 \[ \begin{array}{rrccl@{\quad}l}
A & ::= &
	\bot 
&\mid &
	xA_1\dots A_n & (n \geq 0)
\\ &&&\mid &
	`lx.A & ( A \not= \bot)
\\ &&&\mid &
`m`a.[`b]A & (A \not= `m`g[`d]A', ~ A \not= \bot)
 \end{array} \]

 \item
The relation ${\dirapp} \subseteq 
\lmu\bottom^2 $ is 
the smallest preorder that is the compatible extension of $\bot \dirapp M$\Ignore{:\footnote{Notice that if $\ftAppr_1 \dirapp M_1$, and $\ftAppr_2 \dirapp M_2$, then $\ftAppr_1\ftAppr_2$ need not be an approximant; it is one if $\ftAppr_1 = x\ftAppr_1^1\dots \ftAppr_1^n$, perhaps prefixed with $`m`a.[`b]$.}
 \[ \begin {array}{rcl}
 \bottom &\dirapp& M \\
M \dirapp M' & \Then & `l x . M \dirapp `l x . M' \\
M \dirapp M' & \Then & `m`a.[`b]M \dirapp `m`a.[`b]M' \\
M_1 \dirapp M'_1 \And M_2 \dirapp M'_2 & \Then & M_1M_2 \dirapp M_1'M_2'
 \end {array} \]}{.}

 \item
The set of \emph{approximants} of $M$, $\SetAppr(M)$, is defined
as%
 \[ \begin {array}{rcl}
\SetAppr(M) &\ByDef& \Set{ A \ele \SetAppr \mid \Exists N \ele \lmu \Pred[ M \rtcredbmu N \And A \dirapp N] } .
 \end {array} \]

 \item
\emph{Approximation equivalence} between terms is defined through:
$ \begin{array}{@{}rcl} M \equivA N &\ByDef& \SetAppr(M) = \SetAppr(N). \end{array} $

 \end{enumerate}
 \end{definition}

The relationship between the approximation relation and reduction is characterised by: 


 \begin{lemma} \label{approximation lemma redbmu}
 \begin{enumerate} 

 \firstitem 
If $A \dirapp M$ and $M \rtcredbmu N$, then $A \dirapp N$.

 \item $\lmuHNF$ is a head-normal form if and only if there exists $ A \ele \SetAppr $ such that $ A \dirapp \lmuHNF $ and $A \not= \bot $.

 \end{enumerate}
 \end{lemma}

 \begin{proof}

 \begin{enumerate} 

 \firstitem 
By induction on the structure of approximants.
 \begin{description} 
 \item[$A = \bot$] Trivial, since $\bot \tseq N$.

 \item[$A = xA_1\dots A_n$]
If $ xA_1\dots A_n \tseq M$, then $M \same xM_1\dots M_n$, with $A_i \tseq M_i$ for all $i \ele \n$.
If $M \rtcredbmu N$, then $N = xN_1\dots N_n$ with $M_i \rtcredbmu N_i$, for all $i \ele \n$ (notice that the reduction can take place in many sub-terms, and need not take place in all).
Then, by induction, $A_i \tseq N_i$ for all $i \ele \n$, so $A \tseq N$.

 \item[$A = `lx.A'$, $A' \not= \bot$]
If $ `lx.A' \tseq M$, then $M \same `lx.M'$, with $A' \tseq M'$.
If $M \rtcredbmu N$, then $N = `lx.N'$ with $M' \rtcredbmu N'$.
Then, by induction, $A' \tseq N'$, so $A \tseq N$.

 \item[{$A = `m`a.[`b]A'$, $A' \not= `m`g[`d]A''$, $A' \not= \bot$}]
If $ `m`a.[`b]A' \tseq M$, then $M \same `m`a.[`b]M'$, with $A' \tseq M'$.
Since $A' \not= `m`g[`d]A''$, $M \not= `m`a.[`b]`m`g[`d]M''$, so any reduction in $M$ takes place inside $M'$.
So if $M \rtcredbmu N$, then $N = `m`a.[`b]N'$ with $M' \rtcredbmu N'$.
Then, by induction, $A' \tseq N'$, so $A \tseq N$.

 \end{description}
 
 \item
 \begin{description}  

 \item[\textit{only if}]
By induction on the structure of head-normal forms:
 \begin{description} 

 \item[$\lmuHNF = xM_1\dots M_n$]
Take $A = x\bot \dots \bot$.

 \item[$\lmuHNF = `lx. \lmuHNF'$]
By induction, there exists $A \not= \bot$ such that $A \tseq \lmuHNF'$.
Then $`lx.A \tseq `lx. \lmuHNF'$; notice that, since $A \not= \bot$, also $`lx.A \ele \SetAppr$.

 \item[{$\lmuHNF = \muterm`a.[`b] \lmuHNF', ~\lmuHNF' \not= \muterm`g.[`d] \lmuHNF'' $}]
By induction, there exists $A \not= \bot$ such that $A \tseq \lmuHNF'$.
Then $\muterm`a.[`b]A \tseq \muterm`a.[`b] \lmuHNF'$; notice that, since $A \not= \muterm`g.[`d]A'$ and $A \not= \bottom$, also $\muterm`a.[`b]A \ele \SetAppr$.

 \end{description}

 \item[\textit{if}]
If there exists $ A \ele \SetAppr $ such that $ A \dirapp M $ and $A \not= \bot $, then either:
 \begin{description} 

 \item[$A = xA_1\dots A_n$]
If $ xA_1\dots A_n \tseq M$, then $M \same xM_1\dots M_n$, so $M$ is in head-normal form. 

 \item[$A = `lx.A'$, $A' \not= \bot$]
If $ `lx.A' \tseq M$, then $M \same `lx.M'$, with $A' \tseq M'$.
Since $A' \not= \bot$, by induction $M'$ is in head-normal form, so also $`lx.M'$ is in head-normal form. 

 \item[{$A = `m`a.[`b]A'$, $A' \not= `m`g[`d]A''$, $A' \not= \bot$}]
If $ `m`a.[`b]A' \tseq M$, then $M \same `m`a.[`b]M'$, with $A' \tseq M'$.
Since $A' \not= \bot$, by induction $M'$ is in head-normal form; since $A' \not= `m`g[`d]A''$, also $M' \not= `m`g[`d]M''$, so also $`m`a.[`b]M'$ is in head-normal form. 
\QED

 \end{description}
 \end{description}
 \end{enumerate}
 \end{proof}

The following definition introduces an operation of join on $\lmu\bottom$-terms.

 \begin{definition} [Join, compatible terms]
 \label {join definition}

 \begin{enumerate}

 \firstitem
The partial mapping {\em join}, $\join : \lmu\bottom^2 \rightarrow \lmu\bottom$, is defined by:%
 \[ \begin{array}{r@{\,}c@{\,}lcl}
\bottom \join M \quad \same \quad M &\join& \bottom &\same& M \\
x &\join& x &\same& x \\
(`l x . M) &\join& (`l x . N) &\same& `l x . (M\join N) \\
(`m`a.[`b]M) &\join& (`m`a.[`b]N) &\same& `m`a.[`b](M\join N) \\
(M_1M_2) &\join& (N_1N_2) &\same& (M_1\join N_1)\,(M_2\join N_2) \footnotemark
 \end{array} \]

 \item
If $M \join N$ is defined, then $M$ and $N$ are called {\em compatible}.
 \end{enumerate}

 \end{definition}
\footnotetext{The last alternative in the definition of $\join$ defines the join on applications in a more general way than Scott's, that would state that
 $ \begin{array}{r@{\quad}c@{\quad}l}
(M_1M_2) \join (N_1N_2) &\tseq& (M_1\join N_1) \,(M_2\join N_2),
 \end{array} $ since it is not always sure if a join of two arbitrary terms exists.
Since we will use our more general definition only on terms that are compatible, there is no real conflict.}
It is easy to show that $\join$ is associative and commutative; we will use $\MojM{n}$ for the term $M_1 \join \dots \join M_n$.
Note that $\bottom$ can be defined as the empty join, i.e.\ if $M \same \MojM{0}$, then $M \same \bottom$.

The following lemma shows that the join acts as least upper bound of compatible terms.

 \begin{lemma} \label {tseq lemma}
 \begin{enumerate}
 
 \firstitem \label{join lemma}
If $P \tseq M$, and $Q \tseq M$, then $P \join Q$ is defined, and:%
 \[ \begin{array}{cccc}
P \tseq P \join Q, & Q \tseq P \join Q, & \It{and} & P \join Q \tseq M.
 \end{array} \]

 \item 
If $A_1$, $A_2 \ele \SetAppr (M)$, then $A_1$ and $A_2$ are compatible.

 \end{enumerate}
 \end{lemma}

 \begin{proof}
 \begin{enumerate}
 
 \firstitem
By easy induction on the definition of $\tseq$.

\Comment{
 \begin{enumerate}

 \item \label{bottom case}
If $P \same \bottom$, then $P \join Q \same Q$, so $P \tseq P \join Q, Q \tseq P \join Q$, and $P \join Q \tseq Q \tseq M$.
(The case $Q \same \bottom$ goes similarly.)

 \item
If $P \same x$, then $M \same x$, and either $Q = \bottom$ or also $Q \same x$.
The first case has been dealt with in part~\ref{bottom case}, and for the other: $P \join Q \same x$.
Obviously, $x \tseq x \join x$, $x \tseq x \join x$, and $x \join x \tseq x$.

 \item
If $P \same `l x . N_1$, then $M \same `l x . N$, $N_1\tseq N$, and either $Q = \bottom$ or $Q \same `l x . N_2$.
The first case has been dealt with in part~\ref{bottom case}, and for the other: then $N_2\tseq N$.
Then, by induction, $N_1 \tseq N_1\join N_2$, $N_2 \tseq N_1\join N_2$, and $N_1\join N_2 \tseq N$.
Then also $`l x . N_1 \tseq `l x . N_1\join N_2$, $`l x . N_2 \tseq `l x . N_1\join N_2$, and $`l x . N_1\join N_2 \tseq `l x . N$.
Notice that $`l x . N_1\join N_2 \same (`l x . N_1) \join (`l x . N_2)$.

 \item
If $P \same P_1Q_1$, then $M \same PQ$, $P_1 \tseq P$, $Q_1 \tseq Q$, and either $Q = \bottom$ or $Q \same P_2Q_2$.
The first case has been dealt with in part~\ref{bottom case}, and for the other: then $P_2 \tseq P$, $Q_2 \tseq Q$.
By induction, we know $P_1 \tseq P_1 \join P_2$, $P_2 \tseq P_1 \join P_2$, and $P_1 \join P_2 \tseq P$, as well as $Q_1 \tseq Q_1 \join Q_2$, $Q_2 \tseq Q_1 \join Q_2$, and $Q_1 \join Q_2 \tseq Q$.
Then also $P_1Q_1 \tseq (P_1\join P_2)(Q_1\join Q_2)$, $P_2Q_2 \tseq (P_1\join P_2)(Q_1\join Q_2)$, and $(P_1\join P_2)(Q_1\join Q_2) \tseq PQ$.
Notice that $(P_1\join P_2)(Q_1\join Q_2) \same (P_1Q_1)\join (P_2Q_2)$.%

 \end{enumerate}
}

 \item
If $A_1$, $A_2 \ele \SetAppr (M)$, then there exist $N_1$, $N_2$ such that $M \rtcredbmu N_i $ and $A_i \tseq N_i$, for $i = 1,2$.
Since $\redbmu$ is confluent, there exists $P$ such that $N_i \rtcredbmu P$; then by Lem.\,\ref{approximation lemma redbmu}, also $A_i \tseq P$, for $i = 1,2$. 
Then, by part~\ref{join lemma}, $A_1$ and $A_2$ are compatible.\qed

 \end{enumerate}
 \end{proof}
We can also define $\Sem{M} = \sqcup \, \Set{ A \mid A \ele \SetAppr(M) } $ (which by the previous lemma is well defined); then $\Sem{`.}$ corresponds to (a $\lmu$ variant of) B\"ohm trees \cite{Bohm'68,Barendregt'84}.

As is standard in other settings, interpreting a $\lmu$-term $M$ through its set of approximants $\SetAppr(M)$ gives a semantics.

 \begin{theorem} [Approximation semantics for $\lmu$] \label{approx seman lmux}
If $M \eqbmu N$, then $M \equivA N$.
 \end{theorem}
\vspace*{-2mm}
 \begin{Proof}
By induction on the definition of $\eqbmu$, of which we only show the case $M \rtcredbmu N$.

 \begin {description} 

 \item[$\SetAppr(M) \subseteq \SetAppr(N)$]
If $A \ele \SetAppr(M)$, then there exists $L$ such that $M \rtcredbmu L$ and $A \tseq L$.
Since $\redbmu$ is Church-Rosser, there exists $R$ such that $L \rtcredbmu R$ and $N \rtcredbmu R$, so also $ M \rtcredbmu R$.
Then by Lem.\,\ref{approximation lemma redbmu}, $A \tseq R$, and since $N \rtcredbmu R$, we have $A \ele \SetAppr(N)$.

 \item[$\SetAppr(N) \subseteq \SetAppr(M)$]
If $ A \ele \SetAppr(N)$, then there exists $L$ such that $N \rtcredbmu L$ and $A \tseq L$.
But then also $M \rtcredbmu L$, so $ A \ele \SetAppr(M)$. \QED

 \end {description}
 \end{Proof}
The reverse implication of this result does not hold, since terms without head-normal form (which have only $\bot$ as approximant) are not all related by reduction, so approximation semantics is not fully abstract.

 \section{The approximation and head normalisation results for \texorpdfstring{$\TurnlmuS$}{}}
In this section we will show an approximation result, i.e.\ for every $M$, $`G$, $`S$, and $`D$ such that $\derlmuS `G |- M : `S | `D $, there exists an $A \ele \SetAppr(M)$ such that $\derlmuS `G |- A : `S | `D $.
From this
, the well-known characterisation of (head-)normalisation of $\lmu$-terms using intersection types follows easily, i.e.\ all terms having a (head) normal form are typeable in $\TurnlmuS$ (with a type without $`w$-occurrences).
Another result is the well-known characterisation of strong normalisation of typeable $\lmu$-terms, i.e.\ all terms, typeable in $\TurnlmuS$ without using the rule $(\inter)$ with $I = \emptyset$, are strongly normalisable.

First we give some auxiliary definitions and results.

The rules of the system $\TurnlmuS$ are generalised to $\lmu\bottom$; therefore, if $\bottom$ occurs in a term $M$ and $\D \dcol \derlmuS `G |- M : `S | `D $, in that derivation $\bottom$ has to appear in a position where the rule $(\inter)$ is used with $I = \emptyset$, i.e., in a sub-term typed with $`w$.
Notice that $`l x . \bottom$, $\bottom M_1\dots M_n$, and $`m`a.[`b]\bottom$ are typeable by $`w$ only.

First we show that $\TurnlmuS$ is closed for $\tseq$.
 
 \begin{lemma} \label {closed for tseq}
$\derlmuS `G |- M : `S | `D $ and $ M \tseq N $ then $ \derlmuS `G |- N : `S | `D $.
 \end{lemma}

 \begin{proof}
By easy induction on the definition of $\tseq$; the base case, $\bottom \tseq N$, follows from the fact that then $`S = `w$.
\qed

 \end{proof}

Next we define a notion of type assignment that is similar to that of Def.\,\ref{strict type assignment for lmu}, but differs in that it assigns $`w$ \emph{only to the term} $\bottom$.

 \begin{definition} \label {bottom type assignment definition}

$\bottom${\em -type assignment} and $\bottom${\em -derivations} are defined 
as $\TurnlmuS$
, with the exception of: 
{\def \TurnLmuBot{\Turn}%
 \[ \begin{array}[t]{rl@{\dquad}rl}
\Ignore{
(\Ax) : &
\Inf [n \beq 1, i \ele \n]
	{ \derlmuBot `G,x{:}\AoI{\n} |- x : `A_i | `D }
&}{}
(\interBot) : &
\Inf	[n = 0 \Or n \geq 2]
	{\derlmuBot `G |- M_i : `A_i | `D
	\quad
	(\Forall i \ele \n)
	}{ \derlmuBot `G |- \MojM{n} : \AoI{\n} | `D }
\Ignore{
 \\ [5mm]
(\Abs) : &
\Inf	[x \notele `G]
	{ \derlmuBot `G, x{:}`S |- M : `C\arr `y | `D
	}{ \derlmuBot `G |- `l x . M : `S\prod`C\arr `y | `D }
&
(`m) : &
\Inf	[`a \notele `D, `C \seqS \cont{D}]
	{\derlmuBot `G |- M : \cont{D}\arr `y | `a{:}`C,`D
	}{\derlmuBot `G |- `m`a.[`a]M : `C\arr `y | `D }
 \\ [5mm]
(\App) : &
\Inf	{ \derlmuBot `G |- M : `S\prod`C\arr `y | `D
	 \quad
	 \derlmuBot `G |- N : `S | `D
	}{ \derlmuBot `G |- MN : `C\arr `y | `D }
&
(`m') : &
\Inf	[<20mm>{`b \not= `a \notele `D, \\ `C' \seqS \cont{D}}]
	{\derlmuBot `G |- M : \cont{D}\arr `y | `a{:}`C,`b{:}`C',`D
	}{\derlmuBot `G |- `m`a.[`b]M : `C\arr `y | `b{:}`C',`D }
}{}
 \end{array} \]}
We write $\derlmuBot `G |- M : `S | `D $ if this statement is derivable using a $\bottom$-derivation.

 \end{definition}
Notice that, by rule $(\interBot)$, $\derlmuBot `G |- {\bottom} : `w | `D $, and that this is the only way to assign $`w$ to a term.
Moreover, in that rule, the terms $M_j$ need to be compatible (otherwise their join would not be defined).

 \begin{lemma} \label {comparison}
 \begin{enumerate} 

 \firstitem \label {bottom to full}
If $\D \dcol \derlmuBot `G |- M : `S | `D $, then $\D \dcol \derlmuS `G |- M : `S | `D $.

 \item \label {full to bottom}
If $\D \dcol \derlmuS `G |- M : `S | `D $, then there exists $M' \tseq M$ such that $\D \dcol \derlmuBot `G |- M' : `S | `D $.

 \end{enumerate}
 \end{lemma}

 \begin{proof}
 \begin{enumerate} 

 \firstitem
By induction on the structure of derivations in $\TurnLmuBot$. \Ignore{}{We only show:}
 \begin{description} 

\Ignore{
 \item [$\Ax$]
Immediate.
}{}

 \item [$\interBot$]
Then 
$`S = \AoI{\n}$, $M = \MojM{n}$, and, for every $\iotn$, $\derlmuBot `G |- M_i : `A_i | `D $.
Then, by induction, for every $\iotn$, $\derlmuS `G |- M_i : `A_i | `D $.
Since, by Lem.\,\ref {tseq lemma}, $M_i \tseq M$ for all $\iotn$, by Lem.\,\ref {closed for tseq}, for every $\iotn$, $\derlmuS `G |- M : `A_i | `D $, so by $(\inter)$, $\derlmuS `G |- M : \AoI{\n} | `D $.

\Ignore{
 \item [$\Abs$]
Then $M \same `l x . N$, and $`S = `T\prod `C\arr `y$, and $\derlmuBot `G, x{:}`T |- N : `C\arr `y | `D $.
Then, by induction, $\derlmuS `G,x{:}`T |- N : `C\arr `y | `D $, so by $(\Abs)$, $\derlmuS `G |- `l x . N : `T\prod`C\arr `y | `D $.

 \item [$\App$]
Then $M \same PQ$, $`S = `C\arr`y$, and there exists $`T$ such that $\derlmuBot `G |- P : `T\prod`C\arr `y | `D $, and $\derlmuBot `G |- Q : `T | `D $.
Then, by induction, $\derlmuS `G |- P : `T\prod`C\arr `y | `D $, and $\derlmuS `G |- Q : `T | `D $, so by $(\App)$, $\derlmuS `G |- PQ : `S | `D $.

 \item[$`m$]
Then $M \same `m`a.[`a]N$, $`S = `C\arr `y$, and $\derlmuBot `G |- N : \cont{D}\arr `y | `a{:}`C,`D $ with $`C \seqS \cont{D}$.
By induction, $\derlmuS `G |- N : \cont{D}\arr `y | `a{:}`C,`D $, so by rule $(`m)$, also $\derlmuS `G |- `m`a.[`a]N : `C\arr `y | `D $.

 \item[$`m'$]
Then $M \same `m`a.[`b]N$, $`S = `C\arr `y$, $`D = `b{:}`C',`D'$, and $\derlmuBot `G |- N : \cont{D}\arr `y | `a{:}`C,`b{:}`C',`D' $ with $`C' \seqS \cont{D}$.
By induction, $\derlmuS `G |- N : \cont{D}\arr `y | `a{:}`C,`b{:}`C',`D' $, so by rule $(`m')$, also $\derlmuS `G |- `m`a.[`b]N : `C\arr `y | `b{:}`C',`D' $.
}{}

 \end{description}

\Ignore{}{
All other cases follow by straightforward induction.}

 \item
By induction on the structure of derivations in $\TurnlmuS$. \Ignore{}{We only show:}
 \begin{description} 

\Ignore{
 \item [$\Ax$]
Immediate.
}{}

 \item [$\inter$]
Then $`S = \AoI{\n}$ and, for every $\iotn$, $\derlmuS `G |- M : `A_i | `D $; by induction, for every $\iotn$ there exists $M_i \tseq M$ such that $\derlmuBot `G |- M_i : `A_i | `D $ (notice that then these $M_i$ are compatible).
Then, by rule $(\interBot)$, we have $\derlmuBot `G |- \MojM{n} : `A_i | `D $.
Notice that, by Lem.\,\ref {tseq lemma}, $\MojM{n} \tseq M$.

\Ignore{
 \item [$\Abs$]
Then $M \same `l x.P$, and $`S = `T\prod `C\arr `y$, and $\derlmuS `G,x{:}`T |- P : `C\arr `y | `D $.
So, by induction, there exists $P' \tseq P$ such that $\derlmuBot `G, x{:}`T |- P' : `C\arr `y | `D $.
Then we obtain $\derlmuBot `G |- `l x.P' : `T\prod`C\arr `y | `D $ by rule $(\Abs)$. 
Notice that $`l x.P' \tseq `l x.P$.

 \item [$\App$]
Then $M \same PQ$, $`S = `C \arr `y$, and there is a $`T$ such that $\derlmuS `G |- P : `T\prod`C\arr `y | `D $, and $\derlmuS `G |- Q : `T | `D $.
Then, by induction, there are $P' \tseq P$, and $Q' \tseq Q$, such that $\derlmuBot `G |- P' : `T\prod`C\arr `y | `D $, and $\derlmuBot `G |- Q' : `T | `D $.
Then, by $(\App)$, $\derlmuBot `G |- P'Q' : `S | `D $.
Notice that $P'Q' \tseq PQ$.

 \item[$`m$]
Then $M \same `m`a.[`a]N$, $`S = `C\arr `y$, and $\derlmuS `G |- N : \cont{D}\arr `y | `a{:}`C,`D $ with $`C \seqS \cont{D}$.
By induction, there exists $N' \tseq N$, and $\derlmuBot `G |- N' : \cont{D}\arr `y | `a{:}`C,`D $, so by rule $(`m)$, also $\derlmuBot `G |- `m`a.[`a]N' : `C\arr `y | `D $.
Notice that $ `m`a.[`a]N' \tseq `m`a.[`a]N $.

 \item[$`m'$]
Then $M \same `m`a.[`b]N$, $`S = `C\arr `y$, $`D = `b{:}`C',`D'$, and $\derlmuS `G |- N : \cont{D}\arr `y | `a{:}`C,`b{:}`C',`D' $ with $`C' \seqS \cont{D}$.
By induction, there exists $N' \tseq N$, and $\derlmuBot `G |- N' : \cont{D}\arr `y | `a{:}`C,`b{:}`C',`D' $, so by rule $(`m')$, also $\derlmuBot `G |- `m`a.[`b]N' : `C\arr `y | `b{:}`C',`D' $.
Notice that $ `m`a.[`b]N' \tseq `m`a.[`b]N $.
\qed
}{}

 \end{description}
\Ignore{}{All other cases follow by straightforward induction.\qed}

 \end{enumerate}
 \end{proof}
Notice that, since $M'$ need not be the same as $M$, the second derivation in part~\ref {full to bottom} is not exactly the same; however, it has the same structure in terms of applied derivation rules.

Using Thm.\,\ref {FSCD lemma}\ref{typeable implies SN} and Lem.\,\ref{comparison}, as for the BCD-system (see \cite{Ronchi-Venneri'84}) and the system of \cite{Bakel-TCS'95}, the relation between types assignable to a $\lmu$-term and those assignable to its approximants can be formulated as:

 \begin{theorem} [Approximation] \label{approximation result for TurnlmuS}
$\derlmuS `G |- M : `S | `D \Iff \Exists A \ele \SetAppr(M) \Pred[ \derlmuS `G |- A : `S | `D ] $.
 \end{theorem}
\vspace*{-2mm}
 \begin{Proof}
 \begin{description} 

 \item [$\Then$]
If $\D \dcol \derlmuS `G |- M : `S | `D $, then, by Thm.\,\ref {FSCD lemma}\ref{typeable implies SN}, $\SN(\D)$.
Let $\D' \dcol \derlmuS `G |- N : `S | `D $ be a normal form of $\D$ with respect to $\derred$, then by Thm.\,\ref{FSCD lemma}\ref {relation bred derred}, $M \rtcbred N$ and, by Lem.\,\ref{comparison}\sk\ref{full to bottom}, there exists $N' \tseq N$ such that $\D' \dcol \derlmuBot `G |- N' : `S | `D $.
So, in particular, $N'$ contains no redexes (no redexes typed with a type different form $`w$ since $\D'$ is in normal form, and none typed with $`w$ since only $\bottom$ can be typed with $`w$), so $N' \ele \AppNF$, and therefore $N' \ele \SetAppr(M)$.

 \item [$\If$]
Let $A \ele \SetAppr(M)$ be such that $\derlmuS `G |- A : `S | `D $.
Since $A \ele \SetAppr(M)$, there exists $N$ such that $M \rtcredbmu N$ and $A \tseq N$.
Then, by Lem.\,\ref {closed for tseq}, $\derlmuS `G |- N : `S | `D $, and, by Thm.\,\ref{FSCD lemma}\ref {closed for eq}, also $\derlmuS `G |- M : `S | `D $.
\qed

 \end{description}
 \end{Proof}

Using 
this last result, the characterisation of head-normalisation becomes easy to show.

 \begin{theorem} [Head-normalisation] \label{characterisation of head normal form}
There exists $`G$, $`A$, and $`D$ such that $\derlmuS `G |- M : `A | `D $, if and only if $M$ has a head normal form.
 \end{theorem}
\vspace*{-2mm}
 \begin{Proof}
 \begin{description} 

 \item [\textit{only if}]
If $\derlmuS `G |- M : `A | `D $, then, by Thm.\,\ref {approximation result for TurnlmuS}, there exists an $A \ele \SetAppr(M)$ such that $ \derlmuBot `G |- A : `A | `D $.
Then, by Def.\,\ref {approximant definition}, there exists $N$ such that $M \rtcredbmu N$ and $A \tseq N$.
Since $`A \not= `w$, $A \not\same \bottom$, so we know that $A$ is either $xA_1\dots A_n$ $(n \geq 0)$, $`l x . A'$, or $`m`a.[`b]A'$ with $A' \not= `m`g.[`d]A''$.
Since $A \tseq N$, $N$ is either $xM_1\dots M_n$ $(n \geq 0)$, $`l x . P$, or $`m`a.[`b]P$ with $P \not= `m`g.[`d]Q$.
Then $N$ is in head-normal from and $M$ has a head-normal form.

 \item [\textit{if}]
If $M$ has a head-normal form, then there exists $N$ such that $M \rtcredbmu N$ and either:
 \begin{description} 

 \item [$N \same xM_1\dots M_n$]
Take $`G = x{:}`w \prod \dots \prod `w \prod `W \arr `y$ (with $n$ times $`w$) and $`A = `W\arr `y$.

 \item [$N \same `l x . P$]
Since $P$ is in head-normal form, by induction there are $`G'$, $`C$, $`y$, and $`D'$ such that $\derlmuS `G' |- P : `C\arr `y | `D' $.
If $x{:}`S \ele `G'$, take $`G = `G' \Except x$, and $`A = `S\prod`C \arr `y $; otherwise take $`G = `G'$ and $`A = `w\prod`C \arr `y$.
In either case, by rule $(\Abs)$, $\derlmuS `G |- `lx.P : `A | `D' $

 \item [{$N = `m`a.[`a]P$}]
Since $P$ is in head-normal form, by induction there are $`G'$, $`C$, $\cont{D}$, $`y$, and $`D'$ such that $\derlmuS `G' |- P : \cont{D}\arr `y | `a{:}`C,`D' $.
Take $`C' = `C \inter \cont{D}$, then by Thm.\,\ref{FSCD lemma}\sk\ref{closed for seqS} also $\derlmuS `G' |- P : \cont{D}\arr `y | `a{:}`C',`D' $, and since $`C' \seqS \cont{D}$, by rule $(`m)$ we get $\derlmuS `G' |- `m`a.[`a]P : `C'\arr `y | `D' $.

 \item [{$N = `m`a.[`b]P$, with $`a \not= `b$}]
Since $P$ is in head-normal form, by induction there are $`C$, $`C'$, $\cont{D}$ such that $\derlmuS `G' |- P : \cont{D}\arr `y | `a{:}`C,`b{:}`C',`D $ and $`C' \seqS \cont{D}$.
Take $`C'' = `C' \inter \cont{D}$, then by Thm.\,\ref{FSCD lemma}\sk\ref{closed for seqS} also $\derlmuS `G' |- P : \cont{D}\arr `y | `a{:}`C,`b{:}`C'',`D $, and since $`C'' \seqS \cont{D}$ we get $\derlmuS `G' |- `m`a.[`b]P : `C'\arr `y | `b{:}`C'',`D' $ by $(`m')$.

 \end{description}
Notice that in all cases, $\derlmuS `G |- N : `A | `D $, for some $`A$, and by Thm.\,\ref{FSCD lemma}\ref {closed for eq}, $\derlmuS `G |- M : `A | `D $.
\qed

 \end{description}
 \end{Proof}

 \section {Type assignment for (strong) normalisation} \label {omega free}

In this section we show the characterisation of both normalisation  and strong normalisation, for which we first define a notion of derivability obtained from $\TurnlmuS$ by restricting the use of the type assignment rule $(\inter)$ to at least two sub-derivations, thereby eliminating the possibility to assign $`w$ to a term.

 \begin{definition} [SN type assignment] \label{omega free type assignment definition}
 \begin{enumerate}
 \firstitem
We define the $`w$-free types by the grammar:
%
%
 \[ \begin{array}{rcl@{\quad}l@{\quad}l}
`A,`B &::=& `C\arr `y \\
`R,`S,`T &::=& `A_1 \inter \dots \inter `A_n & (n \geq 1) \\
`C,\cont{D} &::=& `W \mid `S\prod `C 
 \end{array} \]

 \item
{\def \TurnlmuSN {\Turn}
 \emph{SN type assignment} is defined using the natural deduction system
of Def.\,\ref{strict type assignment for lmu}, but allowing only $`w$-free types, so restricting rule $(\inter)$ to:%
 \[ \begin{array}[t]{rl}
(\inter) : &
\Inf	[n \geq 2]
	{\derlmuSN `G |- M : `A_i | `D
	\quad
	(\Forall i \ele \n)
	}{ \derlmuSN `G |- M : \AoI{n} | `D }
 \end{array} \] 
}%
We write $\derlmuSN `G |- M : `S | `D $ if this judgement is derivable using this system.

 \end{enumerate}
 \end{definition}
Notice that the only real change in the system compared to $\TurnlmuS$ is that $`w$ is no longer an intersection type, so in rule $(\inter)$, the empty intersection $`w$ is excluded.\footnote{With the aim of the characterisation of strong normalisation, it would have sufficed to only restrict rule $(\inter)$; we restrict the set of types as well in order to be able to characterise normalisation as well.}

The following properties hold:

 \begin{lemma} \label{inte derivable} \label{renaming lemma} \label {w-free properties} \label{thinning} \label{weakening}

 \begin{enumerate} 

 \firstitem \label{een}
If $ `S \seq `T $, then $ `S = \AoI{I} $, $ `T = \BoI{J}$, and for every $j \ele J$ there exists $i \ele I$ such that $`A_i = `B_j$.


 \item \label{drie}
$\derlmuSN `G,x{:}`S |- x : `T | `D $, if and only if $`S \seqS `T$.


 \item \label{free variables}
$\derlmuSN `G |- M : `S | `D \Then \derlmuSN { \{ x{:}`T \ele `G \mid x \ele \FV{M} \} } |- M : `S | { \{ `a{:}`C \ele `D \mid `a \ele \fn(M) \} } $.

 \item \label {w-free more essential basis}
$\derlmuSN `G |- M : `S | `D \And `G' \supseteq `G \And `D' \supseteq `D \Then \derlmuSN `G' |- M : `S | `D' $.

 \item \label {full to omega-free}
$\D \dcol \derlmuSN `G |- M : `S | `D \Implies \D \dcol \derlmuS `G |- M : `S | `D $.

 \end{enumerate}
 \end{lemma}

 \begin{proof}
Straightforward.\qed
\end{proof}

As for $\TurnlmuS$, we can show that $(\seqS)$ is an admissible rule in $\TurnlmuSN$.

 \begin {lemma} \label{SN seqS admissible} \label{SN closed for seqS}
If $ \derlmuSN `G |- M : `S | `D $, and $`G'$, $`T$, and $`D'$ are all $`w$-free and satisfy $`G' \seqS `G$, $`D' \seqS `D$, and $`S \seqS `T$, then $ \derlmuSN `G' |- M : `T | `D' $.
 \end{lemma}

 \begin{Proof}
\Ignore{
 \begin{description} 

 \item[$\Ax$]
Then $M \same x$, $`S = `A$, and there exists $x{:}`R \ele `G$ such that $`R \seqS `A$.
Since $`G' \seqS `G$, there exists $x{:}`R' \ele `G'$ such that $`R' \seqS `R$.
Notice that then $`R' \seqS `T$, and Lem.\,\ref{w-free properties}\ref{drie}, $ \derlmuE `G' |- x : `T | `D' $.

 \item[$\Abs$]
Then $M = `lx.N$, $ `S = `R\prod `C\arr `y$ and $ \derlmuSN `G,x{:}`R |- N : `C\arr `y | `D $.
Since $`R\prod `C\arr `y \seqS `T$, we have $`S = `T$.
Then by induction $ \derlmuSN `G',x{:}`R |- N : `C\arr `y | `D' $, and we get $ \derlmuSN `G' |- `lx.N : `R\prod`C\arr `y | `D' $ by $(\Abs)$.

 \item[$\App$]
Then $M \same PQ$, $`S = `C\arr `y$, and there exists $`R$ such that $ \derlmuSN `G |- M : `R\prod `C\arr `y | `D $ and $ \derlmuSN `G |- N : `R | `D $.
If $`C\arr `y \seqS `T $, then $`T = `C\arr `y$, so by induction $ \derlmuSN `G' |- P : `R\prod `C\arr `y | `D' $ and by rule $(\App)$ we get $ \derlmuSN `G' |- PQ : `C\arr `y | `D' $.

 \item[$`m$]
Then $ M = `m`a.[`a]N $, $`S = `C\arr `y$, and $ \derlmuSN `G |- N : \cont{D}\arr `y | `a{:}`C,`D $ with $`C \seqS \cont{D}$.
Since $`C\arr `y \seqS `T$, in fact $`T = `C\arr `y$, so by induction $ \derlmuSN `G' |- N : \cont{D}\arr `y | `a{:}`C,`D' $.
Then also $ \derlmuSN `G' |- `m`a.[`a]N : `C\arr `y | `D' $ by rule $(`m)$.

 \item[$`m'$]
Then $ M = `m`a.[`b]N $, $`S = `C\arr `y$, $`D = `b{:}`C_0,`D_0$ and $ \derlmuSN `G |- N : \cont{D}\arr `y | `a{:}`C,`b{:}`C_0,`D_0 $ with $`C_0 \seqS \cont{D}$.
Since $`C\arr `y \seqS `T$, in fact $`T = `C\arr `y$.
Since $`D' \seqS `D$, there exist $`C'_0 \seqS `C_0$ and $`D'_0 \seqS `D_0$ such that $`D' = `b{:}`C'_0,`D'_0$ and $`a{:}`C,`b{:}`C'_0,`D'_0 \seqS `a{:}`C,`b{:}`C_0,`D_0$.
Then by induction we have $ \derlmuSN `G' |- N : \cont{D}\arr `y | `a{:}`C,`b{:}`C'_0,`D'_0 $.
Since $`C'_0 \seqS `C_0 \seqS \cont{D}$, we have $ \derlmuSN `G' |- `m`a.[`b]N : `C\arr `y | `D' $ by rule $(`m')$.%

 \item[$\inter$]
Then $`S = \AoI{I}$, and $ \derlmuSN `G |- M : `A_i | `D $, for all $i \ele I$.
Also, by Lemma~\ref{w-free properties}\ref{een} $`T = \AoI{J}$ with $J \subseteq I$, and $ \derlmuSN `G' |- M : `A_j | `D' $, for all $j \ele J$, so by rule $(\inter)$, $ \derlmuSN `G' |- M : `T | `D' $.
\qed

 \end{description}
}{Much the same as the proof for Thm.\,\ref{FSCD lemma}\ref {closed for seqS} in \cite{Bakel-FSCD'16}.\qed}
\end{Proof}


The following lemma shows a (limited) subject expansion result for $\TurnlmuSN$: it states that if a contraction of a redex is typeable, then so is the redex, provided that the operand $N$ is typeable in its own right; since $N$ might not appear in the contractum, we need to assume that separately. 
Notice that we demand that $N$ is typeable in the same contexts as the redex itself; this property would not hold once we consider contextual closure (in particular, when the reduction takes place under an abstraction); it might be that free names or variables in $N$ get bound in the context.

 \begin{lemma} 
\label{structural substitution lemma SN}
If $\derlmuSN `G |- M [N{`.}`g/`a] : `T | `g{:}`C,`D $ and $ \derlmuSN `G |- N : `B | `D $, then there exists $`S$ such that $ \derlmuSN `G |- M : `T | `a{:}`S\prod `C,`D $ and $\derlmuSN `G |- N : `S | `D $.
 \end{lemma}

 \begin{Proof}
By nested induction; the outermost is on the structure of types, and the innermost on the structure of terms.
\Ignore{
 \begin{description} 

 \item [$`T = \AoI{I}$]

If $\derlmuSN `G |- M [N{`.}`g/`a] : \AoI{I} | `g{:}`C,`D $, then by rule $(\inter)$,
$\derlmuSN `G |- M [N{`.}`g/`a] : `A_i | `g{:}`C,`D $, for every $i \ele I$.
By induction, for every $i \ele I$ there exists $`S_i$ such that $ \derlmuSN `G |- M : `A_i | `a{:}`S_i\prod `C,`D $ and $\derlmuSN `G |- N : `S_i | `D $.
Take $`S = \SoI{I}$, and notice that then $ `S \seqS `S_i $, for every $i \ele I$.
Then, by Lem.\,\ref{SN seqS admissible}, $ \derlmuSN `G |- M : `A_i | `a{:}`S\prod `C,`D $, and by rule $(\inter)$ we get both $\derlmuSN `G |- N : `S | `D $ and $ \derlmuSN `G |- M : `T | `a{:}`S\prod `C,`D $.

 \item[$`T = `C'\arr `y$]
By induction on 
terms; let $ \derlmuSN `G |- N : `B | `D $ and $  \derlmuSN `G |- M [N{`.}`g/`a] : `C'\arr `y | `g{:}`C,`D $.
}{We only show:}
 \begin {description} 

 \item [$ M\same x $]
Then $ x [N{`.}`g/`a] = x $. Take $`S = `B$, then by Lem.\,\ref{thinning}, also $ \derlmuSN `G |- x : `C'\arr `y | `a{:}`S\prod `C,`D $.

\Ignore{
 \item [$M\same `l y.P $]
Notice that $(`ly.P)[N{`.}`g/`a] = `ly.P[N{`.}`g/`a]$; then there exists $`R$ and $\cont{D}$ such that 
$ \derlmuSN `G,y{:}`R |- P [N{`.}`g/`a] : \cont{D}\arr `y | `g{:}`C,`D $ and $`C' = `R\prod \cont{D}$.
Then by induction there exists $`S$ such that $\derlmuSN `G,y{:}`R |- P : \cont{D}\arr `y | `a{:}`S\prod `C,`D $ and $ \derlmuSN `G |- N : `S | `D $, so we get $ \derlmuSN `G |- `ly.P : `C'\arr `y | `a{:}`S\prod `C,`D $ by $(\Abs)$.

 \item [$ M = PQ $]
Notice that $(PQ)[N{`.}`g/`a] = P[N{`.}`g/`a] \, Q[N{`.}`g/`a]$; then there exists $`R$ such that $\derlmuSN `G |- P\,[N{`.}`g/`a] : `R\prod `C'\arr `y | `g{:}`C,`D $ and $\derlmuSN `G |- Q\,[N{`.}`g/`a] : `R | `g{:}`C,`D $.
Then by induction, there are $`S_1$, $`S_2$ such that $ \derlmuSN `G |- P : `R\prod `C'\arr `y | `a{:}`S_1\prod `C,`D $ and $\derlmuSN `G |- N : `S_1 | `D $, as well as $ \derlmuSN `G |- Q : `R | `a{:}`S_2\prod `C,`D $ and $\derlmuSN `G |- N : `S_2 | `D $.
Take $`S = `S_1 \int `S_2$; notice that then $ `S\prod `C \seqS `S_i\prod `C$, so $ `a{:}`S\prod `C,`D \seqS `a{:}`S_i\prod `C,`D $, for $i = 1,2$.
Then by Lem.\,\ref{SN seqS admissible} both $ \derlmuSN `G |- P : `R\prod `C'\arr `y | `a{:}`S\prod `C,`D $ and $ \derlmuSN `G |- Q : `R | `a{:}`S\prod `C,`D $, and by $(\App)$ we get $ \derlmuSN `G |- PQ : `A | `a{:}`S\prod `C,`D $.
Notice that $\derlmuSN `G |- N : `S | `D $ follows by Lem.\,\ref{inte derivable}.

 \item [{$M\same `m`d.[`d]P$}]
Notice that $ (`m`d.[`d]P)[N{`.}`g/`a] = `m`d.[`d]P[N{`.}`g/`a] $.
Then there exists $\cont{D}$ such that $ \derlmuSN `G |- P[N{`.}`g/`a] : \cont{D}\arr `y | `d{:}`C',`g{:}`C,`D $ and $`C' \seqS \cont{D} $.
Then by induction there exists $`S$ such that $ \derlmuSN `G |- P : \cont{D}\arr `y | `d{:}`C',`a{:}`S\prod `C,`D $ and $ \derlmuSN `G |- N : `S | `D $. 
Then $\derlmuSN `G |- `m`d.[`d]P : `C'\arr `y | `a{:}`S\prod `C,`D $ follows from $(`m)$.

 \item [{$ M\same {`m`d.[`b]P} $, $`d \not= `b$, $`a \not= `b$}]
Notice that $ (`m`d.[`b]P)[N{`.}`g/`a] = `m`d.[`b]P[N{`.}`g/`a] $.
Then there exists $\cont{D},`C''$ such that $ \derlmuSN `G |- P[N{`.}`g/`a] : \cont{D}\arr `y | `d{:}`C',`b{:}`C'',`g{:}`C,`D' $, and $`D = `b{:}`C'',`D'$ and $`C'' \seqS \cont{D} $.
Then by induction there exists $`S$ such that $\derlmuSN `G |- P : \cont{D}\arr `y | `d{:}`C',`b{:}`C'',`a{:}`S\prod `C,`D' $ and $ \derlmuSN `G |- N : `S | `D $.
Then by $(`m')$, $ \derlmuSN `G |- `m`d.[`b]P : `C'\arr `y | `a{:}`S\prod `C,`D $.

 \item [{$ M\same {`m`d.[`a]P} $, $`d \not= `a$}]

Notice that then $ (`m`d.[`a]P)[N{`.}`g/`a] = `m`d.[`g]P[N{`.}`g/`a]N $.
Then $ \derlmuSN `G |- `m`d.[`g]P[N{`.}`g/`a]N : `C'\arr `y | `g{:}`C,`D $, and there exists $\cont{D}$ such that $ \derlmuSN `G |- P[N{`.}`g/`a]N : \cont{D}\arr `y | `d{:}`C',`g{:}`C,`D $ by rule $(`m') $, and $ `C \seqS \cont{D} $.
Then by $(\App)$ there is $`R$ such that $ \derlmuSN `G |- P[N{`.}`g/`a] : `R\prod \cont{D}\arr `y | `d{:}`C',`g{:}`C,`D $ and $ \derlmuSN `G |- N : `R | `D $.
By induction, there exists $`T$ such that $ \derlmuSN `G |- P : `R\prod \cont{D}\arr `y | `d{:}`C',`a{:}`T\prod `C,`D $ and $ \derlmuSN `G |- N : `T | `D $.
Take $`S = `R \inter `T$, then $ `S\prod \cont{D} \seqS `T \prod \cont{D} $; so by Lem.\,\ref{SN seqS admissible} also $ \derlmuSN `G |- P : `R\prod \cont{D}\arr `y | `d{:}`C',`a{:}`S\prod `C,`D $.
Since also $`S \prod `C \seqS `R \prod \cont{D}$, we get $\derlmuSN `G |- `m`d.[`a]P : `C'\arr `y | `a{:}`S\prod `C,`D $ by rule $(`m')$ and $ \derlmuSN `G |- N : `S | `D $ by rule $(\inter)$.
\qed

 \end{description}
}{All other cases follow by induction.\qed}

 \end{description}
 \end{Proof}

To prepare the characterisation of terms by their assignable types, we first prove that a term in $\lmu\bottom$-normal form is typeable without $`w$, if and only if it does not contain $\bottom$.
This forms the basis for the result that all normalisable terms are typeable without $`w$.
Notice that the first result is stated for $\TurnlmuS$.


 \begin{lemma} \label {w-free normal forms}

 \begin{enumerate} 

 \firstitem \label {to normal form}
If $\derlmuS `G |- A : `A | `D $, and $`G$, $`A$, and $`D$ are $`w$-free, then $A$ is $\bottom$-free.

 \item \label {from normal form}
If $A$ is $\bottom$-free, then there are $`G$, $`A$, and $`D$, such that $\derlmuSN `G |- A : `A | `D $.

 \end{enumerate}
 \end{lemma}

 \begin{proof}
By induction on the structure of approximate normal forms.

 \begin{enumerate} 
 \item
 \begin{description} 

 \item [$A \same x$]
Immediate.

 \item [$A \same \bottom$]
Impossible, by inspecting the rules of $\TurnlmuS$.

 \item [$A \same `lx.A'$]
By $(\Abs)$, $`A = `T\prod`C\arr `y$, and $\derlmuS `G,x{:}`T |- A' : `C\arr `y | `D $.
Of course also $`G, x{:}`T$, and $`C\arr `y$ are $`w$-free, so by induction, $A'$ is $\bottom$-free, so also $`lx.A'$ is $\bottom$-free.

 \item [$A \same xA_0\dots A_n$]
Then by $(\App)$ and $(\Ax)$, $\derlmuS `G |- A_i : `S_i | `D $, and $x{:}\BoI{\m} \ele `G$, and, for some $j \ele m$, $`S_1\prod`S_2\prod \dots \prod `S_n\prod`C\arr `y = `B_j$ and $`A = `C\arr `y$.
Since each $`S_i$ occurs in $`B_j$, which occurs in $`G$, all are $`w$-free, so by induction each $A_i$ is $\bottom$-free.
Then also $xA_1\dots A_n$ is $\bottom$-free.

 \item[{$A \same `m`a.[`a]A'$, with $A' \not= `m`g.[`d]A''$}]
Then $`A = `C\arr `y$, and by $(`m)$ there exists $\cont{D}$ such that $`S \seqS \cont{D}$, and $\derlmuS `G |- A' : \cont{D}\arr `y | `a{:}`C,`D $.
Since $`C \seqS \cont{D}$, and $`C$ is $`w$-free, so is $\cont{D}$; then, by induction, $A'$ is $\bottom$-free, so so is $`m`a.[`a]A'$.

 \item[{$A \same `m`a.[`b]A'$, with $`a \not= `b$ and $ A' \not= `m`g.[`d]A''$}]
Then $`A = `C\arr `y$, and by $(`m')$ there exists $\cont{D}$, $\cont{D}'$ such that $`D=`b{:}\cont{D}',`D'$, $\cont{D}' \seqS \cont{D}$, and $\derlmuS `G |- A' : \cont{D}\arr `y | `a{:}`C,`b{:}\cont{D}',`D' $.
Since $\cont{D}' \seqS \cont{D}$, and $\cont{D}'$ is $`w$-free, so is $\cont{D}$; then, by induction, $A'$ is $\bottom$-free, so so is $`m`a.[`b]A'$.

 \end{description}

 \item
 \begin{description} 

 \item [$A \same x$]
Then $\derlmuSN x{:}`W\arr`y |- x : `W\arr`y | `D $.

 \item [$A \same `lx.A'$]
By induction $\derlmuSN `G |- A' : `C\arr `y | `D $.
If $x$ does not occur in $`G$, take a $`w$-free $`T$; otherwise, there exist $x{:}`T \ele `G$ and $`T$ is $`w$-free.
In either case, by $(\Abs)$ we obtain $\derlmuSN `G\Except x |- `lx.A' : `T\prod`C\arr `y | `D $.

 \item [$A \same xA_1\dots A_n$]
By induction there are $\Aotn$ such that $\derlmuSN `G |- A_i : `A_i | `D $ for every $\iotn$.
Then $\derlmuSN `G\inter\{x{:}`A_1\prod \dots \prod`A_n\arr `y\} |- xA_1\dots A_n : `y | `D $.

 \item[{$A \same `m`a.[`b]A'$, with $A' \not= `m`g.[`d]A''$}]
By induction $\derlmuSN `G |- A' : `C\arr `y | `b{:}\cont{D},`D $ for some $`G$, $`D$, $`C$, $\cont{D}$, and $`y$.
Then Lem.\,\ref{SN seqS admissible}, also $\derlmuSN `G |- A' : `C\arr `y | `b{:}\cont{D}\inter`C,`D $, so by rule $(`m)$, $\derlmuSN `G |- `m`a.[`a]A' : \cont{D}\inter`C\arr `y | `D $.

 \item[{$A \same `m`a.[`b]A'$, with $`a \not= `b$ and $ A' \not= `m`g.[`d]A''$}]
By induction $\derlmuSN `G |- A' : `C\arr `y | `a{:}`C',`b{:}\cont{D},`D $ for some $`G$, $`D$, $`C$, $`C'$, $\cont{D}$, and $`y$.
Then also $\derlmuSN `G |- A' : `C\arr `y | `a{:}`C',`b{:}\cont{D}\inter`C,`D $, so by rule $(`m')$, $\derlmuSN `G |- `m`a.[`b]A' : `C'\arr `y | `b{:}\cont{D}\inter`C,`D $.
 \QED

 \end{description}
 \end{enumerate}
 \end{proof}


Now, as also shown in \cite{Bakel-TCS'92}, it is possible to characterise normalisable terms.

 \begin{theorem} [Characterisation of Normalisation] \label{characterisation of normalisation}
There exists $`w$-free $`G$, $`D$, and $`A$ such that $\derlmuS `G |- M : `A | `D $, if and only if $M$ has a normal form.
 \end{theorem}

 \begin{Proof}
 \begin{description} 

 \item [$\Then$]
If $\derlmuS `G |- M : `A | `D $, by Thm.\,\ref {approximation result for TurnlmuS} there exists $A \ele \SetAppr(M)$ such that $\derlmuS `G |- A : `A | `D $.
Since $`G$, $`A$, and $`D$ are $`w$-free, by Lem.\,\ref {w-free normal forms}\ref {to normal form}, this $A$ is $\bottom$-free.
By Def.\,\ref {definition approximate normal forms} there exists $N$ such that $M \rtcredbmu N$ and $A \tseq N$.
Since $A$ contains no $\bottom$, $A \equiv N$, so $N$ is a normal form, so $M$ has a normal form.

 \item [$\If$]
If $N$ is the normal form of $M$, then it is a $\bottom$-free approximate normal form.
By Lem.\,\ref {w-free normal forms}\ref {from normal form} there are $`G$, $`A$, and $`D$ such that $\derlmuSN `G |- N : `S | `D $.
By Lem.\,\ref{w-free properties}\ref {full to omega-free} also $\derlmuS `G |- N : `S | `D $, and by 
 Thm.\,\ref{FSCD lemma}\ref {closed for eq}, $\derlmuS `G |- M : `S | `D $, and $`G$, $`S$, and $`D$ are $`w$-free.%
\qed

 \end{description}
 \end{Proof}

In \cite{Bakel-FSCD'16} it is shown that it is possible to characterise the set of all terms that are strongly normalisable with respect to $\redbmu$, using Thm.\,\ref{FSCD lemma}\ref {typeable implies SN}, and the proof for the property that all terms in normal form can be typed in $\TurnlmuSN$, a property that follows here from Lem.\,\ref {w-free normal forms} (see the proof of the previous result).
Other than that, the proof is identical.

The following lemma shows that $\TurnlmuSN$ 
is closed under 
the expansion of redexes (notice that the result is not stated for arbitrary reduction steps, but only for terms that are proper redexes).

 \begin{lemma} \label {omega expansion} \label{w-free expansion}
 \begin{enumerate} 

 \firstitem If $\derlmuSN `G |- M[ N/x ] : `A | `D $ and $\derlmuSN `G |- N : `B | `D $, then $\derlmuSN `G |- (`lx.M)N : `A | `D $.

 \item \label{second part}
If $\derlmuSN `G |- `m`g.[`g]P[ Q{`.}`g/`b ]Q : `A | `D $ and $\derlmuSN `G |- Q : `B | `D $, then $\derlmuSN `G |- (`m`b.[`b]P)Q : `A | `D $.

 \item \label{third part}
If $\derlmuSN `G |- `m`g.[`d]P[Q{`.}`g/`b] : `A | `D $ (with $`b \not= `d$) and $\derlmuSN `G |- Q : `B | `D $, then $\derlmuSN `G |- (`m`b.[`d]P)Q : `A | `D $.

 \item \label{fourth part}
If $\derlmuSN `G |- `m`a.([`d]P)[`b/`g] : `A | `D $, then $\derlmuSN `G |- `m`a.[`b]`m`g.[`d]P : `A | `D $.

 \end{enumerate}
 \end{lemma}

 \begin{proof}
 \begin{enumerate}

 \firstitem

By nested induction on the structure of types and the structure of terms.
We just show the case $`A = `C\arr `y$.

 \begin {description} 

 \item [$ M\same x $]
Then $M[N/x] = N$ and $ \derlmuSN `G |- N : `C\arr`y | `D $.
We have $ \derlmuSN `G,x{:}`C\arr`y |- x : `C\arr`y $ by rule $(\Ax)$.
Then $ \derlmuSN `G |- `lx.x : `B\prod `C\arr `y | `D $ by $(\Abs)$ and $\derlmuSN `G |- (`lx.x)N : `C\arr `y | `D $ by rule $(\App)$.

 \item [$ M\same y\not= x $]
We have $y[N/x] \same y$ and $x \notele \FV{y}$; by Lem.\,\ref{thinning}, $ \derlmuSN `G,x{:}`B |- y : `C\arr `y $.
Then $ \derlmuSN `G |- `lx.y : `B\prod `C\arr `y | `D $ by $(\Abs)$ and $\derlmuSN `G |- (`lx.y)N : `C\arr `y | `D $ by rule $(\App)$.

 \item [$M\same `l y.M' $]
If $ \derlmuSN `G,x{:}`S |- `ly.M : `C\arr `y | `D $, then there exist $`S$, $`C'$, $`R$ such that $ \derlmuSN `G,x{:}`S,y{:}`R |- M : `C'\arr `y | `D $ and $`C = `R\prod `C' $.
Then by induction, we get $ \derlmuSN `G,y{:}`R |- (`lx.M')N : `C'\arr `y | `D $.
Then, by rules $(\App)$ and $(\Abs)$, there exists $`T$ such that $ \derlmuSN `G,y{:}`R,x{:}`T |- M' : `C'\arr `y | `D $ and (also using Lem.\,\ref{thinning}) $ \derlmuSN `G |- N : `T | `D $.
But then we get $ \derlmuSN `G |- `lxy.M' : `T\prod `R\prod `C'\arr `y | `D $ by $(\Abs)$, and by rule $(\App)$ also $ \derlmuSN `G |- (`lxy.M')N : `R\prod `C'\arr `y | `D $.

 \item [$ M = PQ $]
Then there exists $`R$ such that 
$\derlmuSN `G |- P\,[N/x] : `R\prod `C\arr `y | `D $ and $\derlmuSN `G |- Q[N/x] : `R | `D $.
Then by induction, both $\derlmuSN `G |- (`lx.P)N : `R\prod `C\arr `y | `D $ and $\derlmuSN `G |- (`lx.Q)N : `R | `D $.
Then by rules $(\Abs)$, $(\App)$, and $(\inter)$ there are $`S$, $`C_i$, $`y_i$, $`T_i$ ($i \ele \n$) such that $`R = \inter_{\n}`C_i\arr`y_i$, and 
$\derlmuSN `G,x{:}`S |- P : `R\prod `C\arr `y | `D $ and $\derlmuSN `G |- N : `S | `D $,  
as well as $\derlmuSN `G,x{:}`T_i |- Q : `C_i\arr `y_i | `D $ and $\derlmuSN `G |- N : `T_i | `D $, for all $i\ele \n$ (notice that, as above, we can assume that $x$ is not free in $N$).

By $(\inter)$ and Lem.\,\ref{SN closed for seqS} we have 
$\derlmuSN `G,x{:}`S\inter \inter_{\n}`T_i |- P : `R\prod `C\arr `y | `D $ and 
$\derlmuSN `G,x{:}`S\inter \inter_{\n}`T_i |- Q : `R | `D $ as well as 
$\derlmuSN `G |- N : `S\inter \inter_{\n}`T_i | `D $.
By $(\App)$ and Lem.\,\ref{SN closed for seqS} we get $\derlmuSN `G,x{:}`S\inter \inter_{\n}`T_i |- PQ : `C\arr `y | `D $ and by $(\Abs)$ $\derlmuSN `G |- `lx.PQ : `S\inter \inter_{\n}`T_i\prod `C\arr `y | `D $.
Then by $(\App)$ we obtain $\derlmuSN `G |- (`lx.PQ)N : `C\arr `y | `D $.

 \item [{$M\same `m`a.[`a]P$}]
Then there exists $\cont{D}$ such that $ \derlmuSN `G |- P[N/x] : \cont{D}\arr`y | `a{:}`C,`D $ with $`C \seqS \cont{D}$, and by induction $ \derlmuSN `G |- (`lx.P)N : \cont{D}\arr`y | `a{:}`C,`D $.
Then by rules $(\Abs)$ and $(\App)$, there exists $`S$ such that $ \derlmuSN `G,x{:}`S |- P : \cont{D}\arr`y | `a{:}`C,`D $ and $ \derlmuSN `G |- N : `S | `D $ (we can assume that $`a$ is not free in $N$).
Then by rule $(`m)$, $ \derlmuSN `G,x{:}`S |- `m`a.[`a]P : `C\arr`y | `D $ and by $(\Abs)$ and $(\App)$ we get $ \derlmuSN `G |- (`lx.`m`a.[`a]P)N : `C\arr`y | `D $.

 \item [{$M\same `m`a.[`b]P$, $`a \not= `b$}]
Then there exists $`C'$, $\cont{D}$ such that $ \derlmuSN `G |- P[N/x] : \cont{D}\arr`y | `a{:}`C,`b{:}`C',`D $ with $\cont{D} \seqS `C$, and by induction $ \derlmuSN `G |- (`lx.P)N : \cont{D}\arr`y | `a{:}`C,`b{:}`C',`D $.
Then by rules $(\Abs)$ and $(\App)$, there exists $`S$ such that $ \derlmuSN `G,x{:}`S |- P : \cont{D}\arr`y | `a{:}`C,`b{:}`C',`D $ and $ \derlmuSN `G |- N : `S | `b{:}`C',`D $ (we can assume that $`a$ is not free in $N$).
Then by rule $(`m')$, $ \derlmuSN `G,x{:}`S |- `m`a.[`b]P : `C\arr`y | `b{:}`C',`D $ and by $(\Abs)$ and $(\App)$ we get $ \derlmuSN `G |- (`lx.`m`a.[`b]P)N : `C\arr`y | `b{:}`C',`D $.

 \end {description}

 \item 
If $ \derlmuSN `G |- `m`g . [`g]P [Q{`.}`g / `b]Q : `A | `D $, then by rule $(`m)$ there are $`C,\cont{D},`y$ such that $A = `C\arr `y$, $`C \seqS \cont{D}$ and
$ \derlmuSN `G |- P [Q{`.}`g / `b]Q : \cont{D}\arr `y | `g{:}`C,`D $.
Then by rule $(\App)$ there exists $`S_1$ such that
$ \derlmuSN `G |- P [Q{`.}`g / `b] : `S_1\prod\cont{D}\arr `y | `g{:}`C,`D $
and $ \derlmuSN `G |- Q : `S_1 | `g{:}`C,`D $.

By Lem.\,\ref{structural substitution lemma SN} there exists $`S_2$ such that $ \derlmuSN `G |- P : `S_1\prod\cont{D}\arr `y | `b{:}`S_2\prod`C,`D $ and $ \derlmuSN `G |- Q : `S_2 | `D $.
Take $`S = `S_1 \inter `S_2$, then $ \derlmuSN `G |- Q : `S | `D $ by Lem.\,\ref{inte derivable}, and $ \derlmuSN `G |- P : `S_1\prod\cont{D}\arr `y | `b{:}`S\prod`C,`D $ by Lem.\,\ref{SN seqS admissible}.
Then $`S\prod`C \seqS `S_1\prod\cont{D}$, so by rule $(`m)$,
$ \derlmuSN `G |- `m`b.[`b]P : `S\prod`C\arr `y | `D $; then, by rule $(\App)$, we get $ \derlmuSN `G |- (`m`b.[`b]P)Q : `A | `D $.

 \item 
If $ \derlmuSN `G |- `m`g . [`d]P [Q{`.}`g / `b] : `A | `D $, there are $`D',`C,`C',\cont{D},`y$ such that $ `A = `C\arr `y$, $`D = `d{:}`C',`D'$, $`C' \seqS \cont{D}$, and $ \derlmuSN `G |- P [Q{`.}`g / `b] : \cont{D}\arr `y | `g{:}`C,`d{:}`C',`D' $.
By Lem.\,\ref{structural substitution lemma SN}, there exists $`S$ such that
$ \derlmuSN `G |- P : \cont{D}\arr `y | `b{:}`S\prod`C,`d{:}`C',`D' $
and $ \derlmuSN `G |- Q : `S | `d{:}`C',`D' $.
By rule $(`m)$, we get
$ \derlmuSN `G |- `m`b.[`d]P : `S\prod`C\arr `y | `d{:}`C',`D' $ and we get $ \derlmuSN `G |- (`m`b.[`d]P)Q : `A | `D $ by rule $(\App)$.

 \item 
We distinguish the following cases (where we assume that distinct identifiers are not equal, and $N = `m`a.([`d]P)[`b/`g]$ and $M = `m`a.[`b]`m`g.[`d]P$):

 \begin {description}

 \item[{$N = `m`a.[`a]P[`a/`g]$}]
By rule $(`m)$, there are $`D',`C,\cont{D},`y$ such that $ \derlmuSN `G |- M[`a/`g] : \cont{D}\arr `y | `a{:}`C,`D $, $`A = `C\arr `y$, and $`C \seqS \cont{D}$.
Then also $ \derlmuSN `G |- P : \cont{D}\arr `y | `a{:}`C,`g{:}`C,`D $.
Then either:

 \begin {description}

 \item[{$M = `m`a.[`a]`m`g.[`g]P$}] 
By rule $(`m)$, $ \derlmuSN `G |- `m`g.[`g]P : `C\arr `y | `a{:}`C,`D $, and again by rule $(`m)$, $ \derlmuSN `G |- `m`a.[`a]`m`g.[`g]P : `C\arr `y | `D $.

 \item[{$M = `m`a.[`a]`m`g.[`a]P$}] 
By rule $(`m')$, $ \derlmuSN `G |- `m`g.[`a]P : `C\arr `y | `a{:}`C,`D $, and by rule $(`m)$ we get $ \derlmuSN `G |- `m`a.[`a]`m`g.[`a]P : `C\arr `y | `D $.

 \end {description}

 \item[{$N = `m`a.[`a]P[`b/`g]$}]
Then $M = `m`a.[`b]`m`g.[`a]P$. 
By rule $(`m)$, there are $`D',`C,`C',\cont{D},`y$ such that $`A = `C\arr `y$, $`D = `b{:}`C',`D'$, $ \derlmuSN `G |- P[`b/`g] : \cont{D}\arr `y | `a{:}`C,`b{:}`C',`D' $, and $`C \seqS \cont{D}$; then also $ \derlmuSN `G |- P : \cont{D}\arr `y | `a{:}`C,`g{:}`C',`b{:}`C',`D' $.
Then we have $ \derlmuSN `G |- `m`g.[`a]P : `C'\arr `y | `a{:}`C,`b{:}`C',`D' $ by rule $(`m')$, and $ \derlmuSN `G |- `m`a.[`b]`m`g.[`a]P : `C\arr `y | `b{:}`C',`D' $ again by rule $(`m')$.

 \item[{$N = `m`a.[`b]P[`b/`g]$}]
By rule $(`m')$, there are $`D',`C,`C',\cont{D},`y$ such that $`A = `C\arr `y$, $`D = `b{:}`C',`D'$, $ \derlmuSN `G |- P[`b/`g] : \cont{D}\arr `y | `a{:}`C,`b{:}`C',`D' $, and $`C' \seqS \cont{D}$.
Then 
$ \derlmuSN `G |- P : `C\arr `y | `a{:}`C,`g{:}`C',`b{:}`C',`D' $, and either:

 \begin {description}

 \item[{$M = `m`a.[`b]`m`g.[`g]P$}] 
By rule $(`m)$ we get $ \derlmuSN `G |- `m`g.[`g]P : `C'\arr `y | `a{:}`C,`b{:}`C',`D' $, and by rule $(`m')$ we get $ \derlmuSN `G |- `m`a.[`b]`m`g.[`g]P : `C\arr `y | `b{:}`C',`D' $.

 \item[{$M = `m`a.[`b]`m`g.[`b]P$}] 
By rule $(`m')$ we get $ \derlmuSN `G |- `m`g.[`b]P : `C'\arr `y | `a{:}`C,`b{:}`C',`D' $ and again by rule $(`m')$ we get $ \derlmuSN `G |- `m`a.[`b]`m`g.[`b]P : `C\arr `y | `b{:}`C',`D' $.

 \end {description}

 \item[{$N = `m`a.[`d]P[`a/`g]$}]
Then $M = `m`a.[`a]`m`g.[`d]P$. 
By rule $(`m')$, there are $`D',`C,\cont{D},`y$ such that $`D = `d{:}`C',`D'$, $`A = `C\arr `y$, $`C' \seqS \cont{D}$, and $ \derlmuSN `G |- P[`a/`g] : \cont{D}\arr `y | `a{:}`C,`d{:}`C',`D' $.
Then also $ \derlmuSN `G |- P : \cont{D}\arr `y | `a{:}`C,`g{:}`C,`d{:}`C',`D' $.
We get $ \derlmuSN `G |- `m`g.[`d]P : `C\arr `y | `a{:}`C,`d{:}`C',`D' $ by rule $(`m')$, and $ \derlmuSN `G |- `m`a.[`a]`m`g.[`d]P : `C\arr `y | `d{:}`C',`D' $ by rule $(`m)$.

 \item[{$N = `m`a.[`d]P[`b/`g]$}]
Then $M = `m`a.[`b]`m`g.[`d]P$. 
By rule $(`m')$, there are $`D',`C,`C',`C'',\cont{D},`y$ such that $`A = `C\arr `y$, $`D = `b{:}`C',`d{:}`C'',`D'$, $`C'' \seqS \cont{D}$, and $ \derlmuSN `G |- P[`b/`g] : \cont{D}\arr `y | `a{:}`C,`b{:}`C',`d{:}`C'',`D' $; then also $ \derlmuSN `G |- P : \cont{D}\arr `y | `a{:}`C,`b{:}`C',`g{:}`C',`d{:}`C'',`D' $.
By rule $(`m')$ we get $ \derlmuSN `G |- `m`g.[`d]P : `C'\arr `y | `a{:}`C,`b{:}`C',`d{:}`C'',`D' $ and we get $ \derlmuSN `G |- `m`a.[`b]`m`g.[`d]P : `C\arr `y | `b{:}`C',`d{:}`C'',`D' $ again by $(`m')$.
\qed

 \end {description}

 \end{enumerate}
 \end{proof}

Thm.\,\ref {SN theorem} below shows that the set of strongly normalisable terms is exactly the set of terms typeable in the intersection system without using the type constant $`w$.
The proof goes by induction on the leftmost outermost reduction path.
First we introduce the notion of leftmost, outer-most reduction.

 \begin{definition}
An occurrence of a redex {\Redex} in a term $M$ is called the {\em leftmost, outermost redex of $M$} ($\Lor{M}$), if:
 \begin{enumerate} 
 \item There is no redex $\Redex'$ in $M$ such that $\Redex'= \Cont[\Redex]$ with $\Cont[-]\not= [-]$ (\emph{outer-most});
 \item There is no redex $\Redex'$ in $M$ such that 
$M = \Cont_0[{ \Cont_1[\Redex'] \, \Cont_2[\Redex] }] $ ({\em leftmost}).
 \end{enumerate}
We write $M \lored N$ is used to indicate that $M$ reduces to $N$ by contracting $\Lor{M}$.

 \end{definition}

The following lemma formulates a subject expansion result for $\TurnlmuSN$ with respect to left-most outer-most reduction.


 \begin{lemma} \label {l.o. lemma}
Assume $M \lored N$, and $\derlmuSN `G |- N : `C\arr `y | `D $; if $\Lor{M} = PQ$ also assume that $\derlmuSN `G_0 |- Q : `B | `D_0 $. 
Then there exists $`G', `D', `C'$ such that $\derlmuSN `G' |- M : `C'\arr `y | `D' $.
 \end{lemma}

 \begin{proof}

We reason by induction on the structure of terms:
 \begin{description}

 \item[$M \same V P_1 \dots P_n$]
We distinguish two cases:
 \begin{enumerate}

 \item
$V$ is a $\redbmu$-redex, and $N \same V' P_1\dots P_n $, where $V'$ is the result of contracting $V$.
From the fact that $\derlmuSN `G |- V' P_1\dots P_n : `C\arr `y | `D $, we know there are $`S_1, \ldots, `S_n$ such that $\derlmuSN `G |- V' : `S_1\prod \dots \prod `S_n\prod `C\arr `y | `D $, and $\derlmuSN `G |- P_i : `S_i | `D $ for all $\iotn$.
Then by Lem.\,\ref{w-free expansion}, $\derlmuSN `G |- V : `S_1\prod \dots \prod `S_n\prod `C\arr `y | `D $, so also $\derlmuSN `G |- V P_1\dots P_n : `C\arr `y | `D $.

 \item
$V\same y$, so there exists $\jotn$ such that $\Lor{M} = \Lor{P_j}$, $P_j \lored P_j'$, and $N \same y P_1\dots P'\dots P_n$.
From $\derlmuSN `G |- y P_1\dots P_j'\dots P_n : `C\arr `y | `D $, we know there are $`S_1, \ldots, `S_n$ such that $\derlmuSN `G |- y : `S_1\prod \dots \prod `S_n\prod `C\arr `y | `D $, and $\derlmuSN `G |- P_i : `S_i | `D $ for all $i \not= j \ele \n$, and $\derlmuSN `G |- P'_j : `S_j | `D $.
Notice that then there exists $y{:}`T \ele `G$ such that $`T \seqS `S_1\prod \dots \prod `S_n\prod `C\arr `y $.

Then, by induction, there are $`G_j$, $`D_j$, and $`B $ such that $\derlmuSN `G_j |- P_j : `B | `D_j $.
Then 
\[ \derlmuSN `G\inter`G_j\inter \Set{ y{:}`S_1\prod \dots `B \dots \prod `S_n\prod `C\arr `y } |- y P_1\dots P_j \dots P_n : `C\arr `y | `D_j \inter `D. \]

 \end{enumerate}

 \item [$ M \same `ly.M' $] 
If $M \lored N$, then $N = `ly.N'$ and $M' \lored N'$.
Then there exists $`S,\cont{D}$ such that $\derlmuSN `G,y{:}`S |- N' : \cont{D}\arr `y | `D $ and $`C = `S\prod\cont{D}$.
By induction, there exists $`G'$, $`D'$, $`S'$, and $\cont{D}'$ such that $\derlmuSN `G',y{:}`S' |- M' : \cont{D}'\arr `y | `D' $.
Then, by rule $(\Abs)$, $\derlmuSN `G' |- `ly.M' : `S'\prod\cont{D}'\arr `y | `D' $.

 \item [{$M \same `m`a.[`a]P$}]
Then $N = `m`a.[`a]Q$ and $P \lored Q$. 
Since $ \derlmuSN `G |- `m`a.[`a]Q : `C\arr `y | `D $, there exists $\cont{D}$ such that $`C \seqS \cont{D}$, and $ \derlmuSN `G |- Q : \cont{D}\arr `y | `a{:}`C,`D_0 $. 
Then by induction there exist $`G'$, $`C'$, $\cont{D}'$, and $`D'$ such that 
 $
\derlmuSN `G' |- Q : \cont{D}'\arr `y | `a{:}`C',`D' .
 $ 
By Lem.\,\ref{SN seqS admissible}
$ \derlmuSN `G' |- Q : \cont{D}'\arr `y | `a{:}`C' \inter \cont{D}',`D' $ 
and then by rule $(`m)$ we get 
$ \derlmuSN `G' |- `m`a.[`b]Q : `C\arr `y | `D' $.

 \item [{$M \same `m`a.[`b]M'$ with $`a \not= `b$}]
Then $N = `m`a.[`b]N'$ and $M' \lored N'$. 
Since $ \derlmuSN `G |- `m`a.[`b]N' : `C\arr `y | `D $, there are $`D_0$, $\cont{E}$, $\cont{D}$ such that $`D = `b{:}\cont{E},`D'_0$, $\cont{E} \seqS \cont{D}$, and $ \derlmuSN `G |- N' : \cont{D}\arr `y | `a{:}`C,`b{:}\cont{E},`D_0 $. 
Then by induction there exist $`G'$, $`C'$, $\cont{E}'$, $\cont{D}'$, and $`D'$ such that 
 $ 
\derlmuSN `G' |- Q : \cont{D}'\arr `y | `a{:}`C',`b{:}\cont{E}',`D' 
 $. 
Then by Lem.\,\ref{SN seqS admissible} we have
$ \derlmuSN `G' |- Q : `C'\arr `y | `a{:}`C',`b{:}\cont{E}'\inter\cont{D}',`D' $ 
and 
$ \derlmuSN `G' |- `m`a.[`b]Q : `C'\arr `y | `b{:}\cont{E}'\inter\cont{D}',`D' $
follows by rule $(`m')$.
\qed
	
 \end{description}
 \end{proof}

We can now show that all strongly normalisable terms are exactly those typeable in $\TurnlmuSN$.

 \begin{theorem} \label {SN theorem}
$ \Exists `G$, $`D$, $`A \Pred[ \derlmuSN `G |- M : `A | `D ] \Iff M$ is strongly normalisable with respect to \hbox{$\redbmu$.}
 \end{theorem}

 \begin{Proof}
 \begin{description} 

 \item [$\Then$]
If $\D \dcol \derlmuSN `G |- M : `A | `D $, then by Lem.\,\ref {w-free properties}\ref {full to omega-free}, also $\D \dcol \derlmuS `G |- M : `A | `D $.
Then, by Thm.\,\ref{FSCD lemma}\ref {typeable implies SN}, $\D$ is strongly normalisable with respect to $\derred$.
Since $\D$ contains no $`w$, all redexes in $M$ correspond to redexes in $\D$, a property that is preserved by derivation reduction (it does not introduce $`w$).
So also $M$ is strongly normalisable with respect to $\redbmu$.

\Comment{
 \item [$\If$]
By induction on the maximum of the lengths of \It{lor}-reduction sequences for a strongly normalisable term $M$ to its normal form (denoted by $\redsize(M)$).

 \begin{enumerate} 
 \item
If $\redsize(M) = 0$, then $M$ is in normal form, and by Lem.\,\ref {w-free normal forms}\ref {from normal form}, there exist $`G,`D$ and $`A$ such that
$\derlmuSN `G |- M : `A | `D $.

 \item
If $\redsize(M)\beq 1$, so $M$ contains a redex, then let $M \lored N$ by contracting the redex $PQ$.
Then $\redsize(N) < \redsize(M)$, and $\redsize(Q) < \redsize(M)$ (since $Q$ is a proper sub-term of a redex in $M$), so by induction, for some $`G$,  $`G'$, $`D$, $`D'$, $`A$, and $`B$, we have $\derlmuSN `G |- M : `A | `D $ and $\derlmuSN `G' |- Q : `B | `D' $.
Then, by Lem.\,\ref {l.o. lemma}, there exist $`G_1$, $`D_1$, $\type{C}$ such that $\derlmuSN `G_1 |- M : \type{C} | `D_1 $.
If the redex is $`m`a.[`b]`m`g.[`d]P$, then $\redsize(`m`a.[`b]`m`g.[`d]P) > \redsize({`m`a.([`d]P)[`b/`g]})$, so the result follows by induction.
\qed

 \end{enumerate}
}

 \item [$\If$]
By induction on the maximum of the lengths of reduction sequences for a strongly normalisable term $M$ to its normal form (denoted by $\redsize(M)$).

 \begin{enumerate} 
 \item
If $\redsize(M) = 0$, then $M$ is in normal form, and by Lem.\,\ref {w-free normal forms}\ref {from normal form}, there exist $`G$, $`D$ and $`A$ such that $\derlmuSN `G |- M : `A | `D $.

 \item
If $\redsize(M)\beq 1$, so $M$ contains a redex, then let $M \lored N$ by contracting the redex $PQ$.
Then $\redsize(N) < \redsize(M)$, and $\redsize(Q) < \redsize(M)$ (since $Q$ is a proper sub-term of a redex in $M$), so by induction, for some $`G$,  $`G'$, $`D$, $`D'$, $`A$, and $`B$, we have $\derlmuSN `G |- M : `A | `D $ and $\derlmuSN `G' |- Q : `B | `D' $.
Then, by Lem.\,\ref {l.o. lemma}, there exist $`G_1$, $`D_1$, $\type{C}$ such that $\derlmuSN `G_1 |- M : \type{C} | `D_1 $.
If the redex is $`m`a.[`b]`m`g.[`d]P$, then $\redsize(`m`a.[`b]`m`g.[`d]P) > \redsize({`m`a.([`d]P)[`b/`g]})$, so the result follows by induction.
\qed

 \end{enumerate}
 \end{description}
 \end{Proof}

\section*{Conclusions}
We have studied a strict version of the intersection type system for $\lmu$ of \cite{Bakel-Barbanera-Liguoro-LMCS'15}.
Using the fact that derivation reduction (a kind of cut-elimination) is strongly normalisable, 
we have shown an approximation theorem, and from that given a characterisation of head normalisation. 
We have also shown that the system without the type constant $`w$ characterises the strongly normalisable terms and that we can characterise normalisation as well.

 \bibliography{references}

\ITRS{}{
\let\proof\oldproof
\let\endproof\oldendproof
\newpage \input{proofsappendix}
}

\Comment{
\newpage
\input{proofsappendix}
 \bibliography{references}
}

 \end{document}
